%% file: arxiv.tex
\documentclass[10pt,letterpaper]{article}
\usepackage[margin=1in,nohead,centering]{geometry}
\usepackage{amsthm,amsmath}
\usepackage{algorithm}
\usepackage[noend]{algpseudocode}
\usepackage[utf8]{inputenc} 
\usepackage{graphicx,color,wrapfig}
\usepackage{amssymb}
\usepackage{amsmath}
\usepackage{url}
\usepackage{xspace}
\usepackage{float}
\usepackage{graphics}
\usepackage{textcomp} 

\usepackage{subcaption}
\usepackage{tikz}
\usepackage{capt-of}

\newcommand{\Gbar}{\overline{G}}
\newcommand{\Vbar}{\overline{V}}
\newcommand{\Ebar}{\overline{E}}

\newtheorem{example}{Example}
\newtheorem{theorem}[example]{Theorem}
\newtheorem{algorithmPeter}[example]{Algorithm}
\newtheorem{definition}[example]{Definition}

\newtheorem{lemma}[example]{Lemma}

\newtheorem{question}[example]{Question}
\newtheorem{proposition}[example]{Proposition}

\chardef\other=12
\def\mdeactivate{%
\catcode`\&=\other   \catcode`\#=\other
\catcode`\%=\other   \catcode`\~=\other
}
\def\mmakeactive#1{\catcode`#1=\active\ignorespaces}

{
\mmakeactive\
\gdef\obeywhitespace{%
  \mmakeactive\^^M %
  \let^^M=\NewLine %
  \aftergroup\removebox %
  \obeyspaces %
}}

\def\NewLine{\par\indent}
\def\removebox{\setbox0=\lastbox}

\def\mverbatim{\par\begingroup\parindent=0em\tt\mdeactivate\obeywhitespace
\catcode`\|=0  %
}

\def\|{|}

\begin{document}

\title{Minimum length RNA folding trajectories}
\author{A.H. Bayegan \and P. Clote\thanks{Correspondence
{\tt clote@bc.edu}.  Research supported in part by
National Science Foundation grant DBI-1262439.}}
\date{Biology Department, Boston College, Chestnut Hill, MA}

\maketitle
\begin{abstract}
{\bf Background:}
Existent programs for RNA folding kinetics, such as
{\tt Kinefold}, {\tt Kinfold} and {\tt KFOLD}, implement the
Gillespie algorithm to generate stochastic folding trajectories
from an initial structure $s$ to a target structure $t$, in which each
intermediate secondary structure is obtained from its predecessor by
the application of a move from a given move set.
The {\tt Kinfold} move set $MS_1$ [resp. $MS_2$] allows the addition or removal
[resp. addition, removal or shift] of a single base pair. Define the
$MS_1$ [resp. $MS_2$] distance between secondary structures $s$ and $t$ to
be the minimum path length to refold $s$ to $t$, where a move from
$MS_1$ [resp. $MS_2$] is applied in each step. The $MS_1$ distance between
$s$ and $t$ is trivially equal to the cardinality of the symmetric difference
of $s$ and $t$, i.e the number of base pairs belonging to one structure
but not the other; in contrast, the computation of $MS_2$ distance is highly
non-trivial.

\noindent
{\bf Results:}
We describe algorithms to compute the shortest $MS_2$ folding 
trajectory between any two given RNA secondary structures. These algorithms
include an optimal integer programming (IP) algorithm, an accurate and efficient
near-optimal algorithm, a greedy algorithm, a branch-and-bound algorithm,
and an optimal algorithm if one allows intermediate structures to contain 
pseudoknots. A 10-fold slower version of our IP algorithm appeared in
WABI 2017; the current version exploits special treatment of closed 2-cycles.

Our optimal IP [resp. near-optimal IP] algorithm
maximizes [resp. approximately maximizes] the number of shifts and
minimizes [resp. approximately minimizes] the number of base pair
additions and removals by applying integer programming to
(essentially) solve the minimum feedback vertex set (FVS) problem
for the RNA conflict digraph, then applies topological sort to tether
subtrajectories into the final optimal folding trajectory.

We prove NP-hardness of the problem to determine the minimum barrier energy over
all possible $MS_2$ folding pathways, and conjecture that
computing the $MS_2$ distance between arbitrary secondary structures is NP-hard.
Since our optimal IP algorithm relies on the FVS, known to be NP-complete 
for arbitrary digraphs, we compare the family of RNA conflict digraphs with
the following classes of digraphs --
planar, reducible flow graph, Eulerian, and tournament --
for which FVS is known to be either polynomial time computable or NP-hard.

\noindent
{\bf Conclusion:}
This paper describes a number of optimal and near-optimal algorithms to
compute the shortest $MS_2$ folding trajectory between any two secondary
structures. 
Source code for our algorithms 
is available at \url{http://bioinformatics.bc.edu/clotelab/MS2distance/}.
\end{abstract}

\section{Background}
\label{section:intro}

RNA secondary structure is known to form a scaffold for tertiary
structure formation \cite{Cho.pnas09}.  Moreover, secondary structure 
can be efficiently predicted with reasonable accuracy by using either
machine learning with stochastic context-free grammars 
\cite{Knudsen.nar03,Schattner.nar05,Sukosd.bb11}, provided that the
training set is sufficiently large and representative, or by using
{\em ab initio} physics-based models with thermodynamics-based algorithms 
\cite{Mathews.jmb99,Lorenz.amb11}. Since the latter approach does not 
depend on any form of homology modeling, it has been successfully used
for synthetic RNA molecular design 
\cite{Zadeh.jcc11,Dotu.nar15,GarciaMartin.nar15},
to predict microRNA binding sites \cite{Rajewsky.ng06},
to discover noncoding RNA genes \cite{Washietl.cpb07},
in simulations to study molecular evolution 
\cite{Borenstein.pnas06,Wagner:robustness,schusterStadler:conformationalEvolution,GarciaMartin.bb16} and in
folding kinetics \cite{flamm,Wolfinger:04a,Senter.jmb14,Dykeman.nar15}.
Software to simulate RNA secondary structure folding kinetics, such as
{\tt Kinfold} and {\tt KFOLD}, implement the Gillespie algorithm to
simulate the moves from one structure to another, for a particular move set.
At the elementary-step resolution, two move sets have extensively been 
studied -- the move set $MS_1$ which allows the addition or removal of
a single base pair, and the move set $MS_2$, which allows the addition,
removal or {\em shift} of a single base pair, where a shift move modifies
only one of the two positions in a base pair, 
as shown in Figure~\ref{fig:shiftMoves}.

In simulation studies related to RNA secondary structure evolution, 
the structural distance between two secondary structures $s,t$ is often
measured by the {\em base pair distance}, denoted $d_{BP}(s,t)$,
defined to be the cardinality of the symmetric difference, 
$| s \bigtriangleup t | = |s-t|+|t-s|$, 
i.e. the number of base pairs belonging to $s$ but not $t$, plus the
number of base pairs belonging to $t$ but not $s$. 
In studies concerning RNA folding kinetics,  the
fast, near-optimal algorithm {\tt RNAtabupath} \cite{Dotu.nar10} and
the much slower, but exact (optimal) 
{\tt Barriers} algorithm \cite{Lorenz.amb11} can be used to determine
$MS_1$ folding trajectories that minimize the {\em barrier energy},
defined as the maximum of the (Turner) free energy difference between 
an intermediate
structure and the initial structure. Thermodynamics-based
software such as
{\tt Kinfold}, {\tt RNAtabupath}, and {\tt KFOLD} use the nearest 
neighbor free energy model \cite{Turner.nar10} whose energy parameters 
are inferred from optical melting experiments. In contrast, the two
theorems below concern
the Nussinov energy model \cite{nussinovJacobson}, 
which assigns $-1$ per base pair and ignores entropy.  Folding 
trajectories $s=s_0,s_1,\ldots,s_m = t$ from $s$ to $t$ may either 
be {\em direct}, whereby each intermediate structure $s_i$ is required 
to contain only base pairs from $s \cup t$, or {\em indirect}, without 
this restriction. 
Note that indirect pathways may be energetically more favorable, 
though longer,
than direct pathways, and that the problem of constructing an 
energetically optimal direct folding pathway is NP-hard. Indeed, the
following theorem is proven in \cite{condonNPcompleteBarrierPathJournal}.

\begin{theorem}[Ma\v{n}uch et al.  \cite{condonNPcompleteBarrierPathJournal}]
\hfill\break
\label{thm:NPhardnessMS1}
With respect to the Nussinov energy model,
it is NP-hard to determine, for given secondary structures $s,t$ and
integer $k$, whether there exists a direct $MS_1$ folding trajectory from
$s$ to $t$ with energy barrier at most $k$.
\end{theorem}

By an easy construction, we can show an analogous result for $MS_2$ folding
pathways. First, we define a {\em direct} $MS_2$ folding pathway from
secondary structure $s$ to secondary structure $t$ to be a folding
pathway $s=s_0,s_1,\ldots,s_n=t$ where each intermediate structure $s_i$
is obtained from $s_{i-1}$ by removing a base pair that belongs to $s$, 
adding a base pair that belongs to $t$, or shifting a base pair belonging
to $s$ into a base pair belonging to $t$.

\begin{theorem}
\label{thm:NPhardnessMS2}
With respect to the Nussinov energy model,
it is NP-hard to determine, for given secondary structures $s,t$ and
integer $k$, whether there exists a direct $MS_2$ folding trajectory from
$s$ to $t$ with energy barrier at most $k$.
\end{theorem}
\begin{proof}
Given secondary structures $s,t$ for an RNA sequence 
${\bf a} = a_1,\ldots,a_n$, without loss of generality 
we can assume that $s,t$
share no common base pair (otherwise, a minimum energy folding
trajectory for $s - (s \cap t)$ and $t-(s \cap t)$ yields a
minimum energy folding trajectory for $s,t$.)
Define the corresponding secondary structures
\begin{align*}
s' &= \{ (2i,2j): (i,j) \in s \}\\
t' &= \{ (2i-1,2j-1): (i,j) \in t \}\\
a'_{2i} &= a_i = a'_{2i-1}  \quad\mbox{ for each $1\leq i \leq n$}\\
{\bf a'} &= a'_1,a'_2,\ldots,a'_{2n}
\end{align*}
In other words, the sequence ${\bf a'} = a_1,a_1,a_2,a_2,\ldots,a_n,a_n$
is obtained by duplicating each nucleotide of ${\bf a}$, and placing each
copy beside the original nucleotide; $s'$ [resp. $t'$] is obtained by
replacing each base pair $(i,j) \in s$ by the base pair
$(2i,2j) \in s'$ [resp. $(2i-1,2j-1) \in t'$. Since there are no
base-paired positions that are shared between $s'$ and $t'$, no shift
moves are possible, thus any direct $MS_2$ folding pathway from $s'$ to
$t'$ immediately yields a corresponding direct $MS_1$ folding pathway from
$s$ to $t$. 
Since the Nussinov energy of any secondary
structure equals $-1$ times the number of base pairs, it follows that
barrier energy of the direct $MS_2$ pathway from $s'$ to $t'$
is identical to that of the corresponding direct $MS_1$ pathway from $s$ 
to $t$. Since $MS_1$ direct barrier energy is an NP-hard problem by 
Theorem \ref{thm:NPhardnessMS1}, it now follows that the $MS_2$ barrier
energy problem is NP-hard.
\end{proof}

Shift moves, depicted in Figure~\ref{fig:shiftMoves}, naturally
model both helix zippering and defect diffusion, depicted in
Figure~\ref{fig:defectDiffusion} and described in \cite{defectDiffusion}.
However, shift moves have rarely been considered in the literature, 
except in the context of folding
kinetics \cite{flamm}. For instance, presumably due to the absence of 
any method to compute $MS_2$ distance, Hamming distance is used as a proxy 
for $MS_2$ distance in the work on molecular evolution of secondary structures 
appearing in \cite{schusterStadler:conformationalEvolution} -- see also
\cite{Wagner.bj14}, where Hamming distance is used to quantify 
structural diversity in defining phenotypic {\em plasticity}.

In this paper, we introduce the first algorithms to compute the $MS_2$
distance between two secondary structures. Although $MS_1$ distance,
also known as base pair distance, is trivial to compute, we conjecture that
$MS_2$ distance is NP-hard, where this problem can be formalized
as the problem
to determine, for any given secondary structures $s,t$
and integer $m$, whether there is an $MS_2$ trajectory 
$s=s_0,s_1,\ldots,s_m = t$ of length $\leq m$.  We describe an
optimal (exact, but possibly exponential time) integer programming (IP) 
algorithm, a fast, near-optimal algorithm, an
exact branch-and-bound algorithm, and a greedy algorithm. 
Since our algorithms involve the 
{\em feedback vertex set} problem for {\em RNA conflict digraphs},
we now provide a bit of background on this problem.

Throughout, we are exclusively interested in {\em directed graphs}, or 
{\em digraphs}, so unless otherwise indicated, all graphs are assumed to 
be directed. Any undefined graph-theoretic concepts can be found in the
monograph by Bang-Jensen and Gutin \cite{bookDigraphdirectedgraph}.
Given a directed graph $G=(V,E)$, a 
{\em feedback vertex set} (FVS) is a subset $V' \subseteq V$ which
contains at least one vertex from every directed 
cycle in $G$, thus rendering $G$ acyclic.  Similarly, a
{\em feedback arc set} (FAS) is a subset $E' \subseteq E$ which
contains at least one directed edge (arc)
from every directed cycle in $G$.  
The FVS [resp. FAS] problem is the problem to determine a minimum size
feedback vertex set [resp.  feedback arc set] which renders $G$ acyclic.
The FVS [resp. FAS] problem can be formulated as a decision problem 
as follows.  Given an integer $k$ and a digraph $G=(V,E)$, determine 
whether there exists a subset $V' \subseteq V$ of size $\leq k$
[resp. $E' \subseteq E$ of size $\leq k$],
such that every directed cycle contains a vertex in $V'$ [resp. an edge in
$E'$].

In Proposition 10.3.1 of \cite{bookDigraphdirectedgraph}, it is proved that
FAS and FVS have the same computational complexity,  within a polynomial factor.
In Theorem 10.3.2 of \cite{bookDigraphdirectedgraph}, it is proved that
the FAS problem is NP-complete -- indeed, this problem appears in the original
list of 21 problems shown by R.M. Karp to be NP-complete \cite{karpNPcomplete}.
Note that Proposition 10.3.1 and Theorem 10.3.2 imply immediately that the FVS
problem is NP-complete. 
In Theorem 10.3.3 of \cite{bookDigraphdirectedgraph}, it is proved that
the FAS problem is NP-complete for tournaments, where a tournament is a
digraph $G=(V,E)$, such that there is a directed edge from $x$ to $y$,
or from $y$ to $x$, for every pair of distinct vertices $x,y \in V$.
In \cite{feedbackarcsetNPcompleteForEulerianDigraphs}, it is proved that
the FAS for Eulerian digraphs is NP-complete, where an Eulerian digraph
is characterized by the property that the in-degree of every vertex equals
its out-degree.
In Theorem 10.3.15  of \cite{bookDigraphdirectedgraph}, it is proved that
FAS can be solved in polynomial time for planar digraphs, a result originally
due to \cite{planarMinimaxArcTheorem}. 
In \cite{feedbackArcProblemReducibleFlowGraphsflowgraph}, a 
polynomial time algorithm is given for the FAS for {\em reducible flow graphs},
a type of digraph that models programs without any GO TO statements
(see \cite{flowgraphUllman} for a characterization of reducible flow graphs).
There is a long history of work on the feedback vertex set and feedback arc
set problems, both for directed and undirected graphs, including results on
computational complexity as well as exact and approximation algorithms for
several classes of graphs -- see the survey \cite{feedbackProblemSurvey}
for an overview of such results.

The plan of the paper is now as follows. In 
Section~\ref{section:ms2distancePseudoknottedStr}, we present 
the graph-theoretic framework for 
our overall approach and describe a simple, fast algorithm to compute the 
{\em pseudoknotted} $MS_2$ distance, or pk-$MS_2$ distance, 
between structures $s,t$. By this we mean the minimum length of an $MS_2$ 
folding trajectory between $s$ and $t$, {\em if} intermediate pseudoknotted 
structures are allowed. We show that the pk-$MS_2$ distance between
$s$ and $t$, denoted by $d_{pk-MS_2}(s,t)$, is approximately equal to one-half
the Hamming distance $d_H(s,t)$ between $s$ and $t$.
This result can be seen as justification, {\em ex post facto}, for
the use of Hamming distance in the investigation of RNA molecular evolution
\cite{schusterStadler:conformationalEvolution}.
In Section~\ref{section:exactIPalgorithm}, we describe an exact integer
programming (IP) algorithm which enumerates all directed cycles, then 
solves the feedback vertex problem for the collection of RNA conflict 
digraphs, as described in Section~\ref{subsection:RNAconflictDigraph}. Our
IP algorithm is not a simple reduction to the feedback vertex set 
(FVS) problem; however, since the complexity of FVS/FAS is known for certain
classes of digraphs, we take initial steps towards the characterization of
RNA conflict digraphs. Our optimal IP algorithm is much faster than a 
branch-and-bound algorithm, but it can be too slow to be practical
to determine $MS_2$ distance between the minimum free energy (MFE) secondary 
structure and a (Zuker) suboptimal secondary structure for some sequences from
the Rfam database \cite{Nawrocki.nar15}. For this reason, in 
Section~\ref{section:nearOptimalAlgorithm} we present
a fast, near-optimal algorithm, and in Section~\ref{section:benchmarking},
we present benchmarking results to compare various algorithms of the paper.

Since we believe that further study of RNA conflict digraphs may lead to a
solution of the question whether $MS_2$ distance is NP-hard, in
Appendix~\ref{section:edgeClassificationAppendix}, all 
types of directed edge that are possible in an RNA conflict digraph are
depicted.
Appendix~\ref{section:pathsOneToFive} presents details on minimum
length pseudoknotted $MS_2$ folding pathways, used to provide a lower
bound in the branch-and-bound algorithm of 
Appendix~\ref{section:branchAndBoundAlgorithm}.
Appendix~\ref{subsection:greedyAlgorithm} presents pseudocode for a 
greedy algorithm.

All algorithms described in this paper have been implemented 
in Python, and are publicly available at
\url{bioinformatics.bc.edu/clotelab/MS2distance}. Our software uses
the function {\tt simple\_cycles(G)} from the software {\tt NetworkX}
\url{https://networkx.github.io/documentation/networkx-1.9/reference/generated/networkx.algorithms.cycles.simple_cycles.html}, and the integer programming
(IP) solver Gurobi Optimizer version 6.0 \url{http://www.gurobi.com, 2014}.

\section{$MS_2$ distance between possibly pseudoknotted structures}
\label{section:ms2distancePseudoknottedStr}

In this section, we describe a straightforward algorithm to determine the
$MS_2$-distance $d_{pk-MS_2}(s,t)$ between any two structures $s,t$ of a
given RNA sequence $a_1,\ldots,a_n$, where 
$d_{pk-MS_2}(s,t)$ is defined to be length of a minimal length trajectory 
$s=s_0,s_1,\ldots,s_m=t$, where intermediate structures $s_i$ may 
contain pseudoknots, but do not contain any base triples. This variant
is called {\em pk-$MS_2$ distance}. Clearly,
the pk-$MS_2$ distance is less than or equal to the $MS_2$ distance.
The purpose of
this section is primarily to introduce some of the main concepts 
used in the remainder of the paper.
Although the notion of secondary structure is well-known,
we give three distinct but equivalent definitions, 
that will allow us to overload 
secondary structure notation to simplify presentation of our algorithms.

\begin{definition}[Secondary structure as set of ordered base pairs]
\label{def:secStr}
Let $[1,n]$ denote the set $\{ 1,2,\ldots,n\}$.
A secondary structure for a given RNA sequence $a_1,\ldots,a_n$
of length $n$ is defined to be a set $s$ of \underline{ordered} pairs
$(i,j)$, with $1 \leq i<j \leq n$, such that
the following conditions are satisfied.
\begin{description}
\item[]
1. {\em Watson-Crick and wobble pairs:}
If $(i,j) \in s$, then
$a_ia_j \in \{ GC,CG,AU,UA,GU,UG \}$.
\item[]
2. {\em No base triples:}
If $(i,j)$ and $(i,k)$ belong to $s$, then $j=k$;
if $(i,j)$ and $(k,j)$ belong to $s$, then $i=k$.
\item[]
3. {\em Nonexistence of pseudoknots:}
If $(i,j)$ and $(k,\ell)$ belong to $s$, then it is not the case that
$i<k<j<\ell$.
\item[]
4. {\em Threshold requirement for hairpins:}
If $(i,j)$ belongs to $s$, then $j-i > \theta$, for a fixed value $\theta\geq 0$; i.e. there must be
at least $\theta$ unpaired bases in a hairpin loop. Following standard
convention, we set $\theta=3$ for steric constraints.
\end{description}
\end{definition}

Without risk of confusion, it will be convenient to overload the concept 
of secondary structure $s$ with two alternative, equivalent notations, 
for which context will determine the intended meaning.
\begin{definition}[Secondary structure as set of unordered base pairs]
\label{def:secStrBis}
A secondary structure $s$ for the RNA sequence $a_1,\ldots,a_n$
is a set of \underline{unordered} pairs
$\{i,j\}$, with $1 \leq i,j \leq n$, such that the corresponding set
of ordered pairs 
\begin{align}
\label{eqn:secStrUnorderedBasePair}
\{ i,j\}_{<} \stackrel{\mbox{\small def}}{=} (\min(i,j),\max(i,j))
\end{align}
satisfies Definition~\ref{def:secStr}.
\end{definition}

\begin{definition}[Secondary structure as an integer-valued function]
\label{def:secStrTer}
A secondary structure $s$ for $a_1,\ldots,a_n$ is a function 
$s: [1,\ldots,n] \rightarrow [0,\ldots,n]$, such that
$\Big\{ \{i,s[i]\}_< : 1 \leq i \leq n, s[i] \ne 0 \Big\}$ 
satisfies Definition~\ref{def:secStr}; i.e.
\begin{equation}
\label{eqn:secStrFunction}
s[i] =  
\left\{ \begin{array}{ll}
0&\mbox{if $i$ is unpaired in $s$}\\
j&\mbox{if $(i,j) \in s$ or $(j,i) \in s$}\\
\end{array} \right.
\end{equation}
\end{definition}

\begin{definition}[Secondary structure distance measures]
\label{def:secStrMetrics}
Let $s,t$ be secondary structures of length $n$. Base pair distance is
defined by equation~(\ref{eqn:basePairDistance}) below, and Hamming distance is
defined by equation~(\ref{eqn:HammingDistance}) below.
\begin{align}
\label{eqn:basePairDistance}
d_{BP}(s,t) &= 
|\{ (x,y): \left((x,y) \in s \land (x,y) \not\in t \right)
\lor \left((x,y) \in t \land (x,y) \not\in s \right) \} | \\
\label{eqn:HammingDistance}
d_H(s,t) &= | \{ i \in [1,n] : s[i] \ne t[i] \} |
\end{align}
\end{definition}

Throughout this section, the term {\em pseudoknotted structure}
is taken to mean a
set of ordered pairs [resp. unordered pairs resp. function], which satisfies
conditions 1,2,4 (but not necessarily 3) of Definition~\ref{def:secStr}.
Given structure $s$  on RNA sequence
$\{ a_1, \ldots,a_n \}$, we say that a position $x \in [1,n]$ is 
{\em touched} by $s$ if $x$ belongs to a base pair of $s$, or equivalently
$s[x]\ne 0$.  For possibly pseudoknotted structures $s,t$ on 
$\{ a_1,\ldots,a_n \}$, we partition the set $[1,n]$ into disjoint sets
$A$,$B$,$C$,$D$ as follows.
Let $A$ be the set of positions that are touched by both $s$ and $t$, yet
do not belong to the same base pair in $s$ and $t$, so
\begin{align}
\label{eqn:defA}
A &= \{ i \in [1,n]: s[i] \ne 0, t(i) \ne 0, s[i] \ne t[i] \}
\end{align}
Let $B$ be the set of positions that are touched by either $s$ or $t$, but 
not by both, so
\begin{align}
\label{eqn:defB}
B &= \{ i \in [1,n]: \left( s[i] \ne 0, t[i]=0 \right)
\lor \left( s[i] = 0, t[i] \ne 0 \right) \}
\end{align}
Let $C$ be the set of positions touched by neither $s$ nor $t$, so
\begin{align}
\label{eqn:defC}
C &= \{ i \in [1,n]: s[i]=0=t[i] \}
\end{align}
Let $D$ be the set of positions that belong to the same base pair in both 
$s$ and $t$, so
\begin{align}
\label{eqn:defD}
D &= \{ i \in [1,n]: s[i] \ne 0, t[i] \ne 0, s[i]=t[i] \}
\end{align}

We further partition $A \cup B$ into a set of maximal paths and cycles, in the 
following manner. Define an undirected, vertex-colored and
edge-colored graph $G=(V,E)$, whose 
vertex set $V$ is equal to the set $A \cup B$ of positions that 
are touched by either $s$ or $t$, but not by a common base pair
in $(s \cap t)$, and whose edge set
$E= (s-t) \cup (t-s) = (s \cup t) - (s \cap t)$ 
consists of undirected edges between positions that are base-paired
together. Color edge $\{ x,y \}$ {\em green} if the base pair 
$(x,y) \in s-t$ and {\em red} if
$(x,y) \in t-s$. Color vertex $x$ {\em yellow} if $x$ is incident to both
a red and green edge, {\em green} if $x$ is incident to a green edge, but not
to any red edge, {\em red} if $x$ is incident to a red edge, but not to any
green edge. 
The connected components of $G$ can be classified into 4 types
of (maximal) paths  and one type of cycle (also called path of type 5):
type 1 paths have two green end nodes,
type 2 paths have a green end node $x$ and a red end node $y$, where $x<y$,
type 3 paths have a red end node $x$ and a green end node $y$, where $x<y$,
type 4 paths have two red end nodes, and type 5 paths (cycles) have no
end nodes. These are 
illustrated in Figure~\ref{fig:redBluePathsCycles}. Note that all
nodes of a cycle and interior nodes of paths of type 1-4 are yellow,
while end nodes (incident to only one edge) are either green or red.
If $X$ is a connected component of $G$, then define the
{\em restriction} of $s$ [resp. $t$] to $X$, denoted by
$s \upharpoonright{X}$ [resp. $t \upharpoonright{X}$], to be the
set of base pairs $(i,j)$ in $s$ [resp. $t$] such that $i,j \in X$.
With this description, most readers will be able to determine
a minimum length pseudoknotted folding pathway from
$s \upharpoonright{X}$ to $t \upharpoonright{X}$, where $X$ is a 
connected component of $G$. For instance, if $X$ is a path of type 2 or 3,
then a sequence of shift moves transforms
$s \upharpoonright{X}$ into $t \upharpoonright{X}$, beginning with a
shift involving the terminal green node. Further details can be found in
Appendix~\ref{section:pathsOneToFive}.
The formal definitions given below
are necessary to provide a careful proof
of the relation between Hamming
distance and pseudoknotted $MS_2$ distance,  also found in
Appendix~\ref{section:pathsOneToFive}.

\begin{definition}
\label{def:equivRelation}
Let $s,t$ be (possibly pseudoknotted) structures on the RNA sequence 
$a_1,\ldots,a_n$.
For $i,j \in [1,n]$, define $i \sim j$ if $s[i]=j$ or $t[i]=j$, and let
$\equiv$ be the reflexive, transitive closure of $\sim$. Thus $i \equiv j$
if $i=j$, or $i = i_1 \sim i_2 \sim \cdots \sim i_m = j$ for any $m \geq 1$.
For $i \in [1,n]$, let $[i]$ denote the equivalence class of $i$, i.e.
$[i] = \{j \in [1,n]: i \equiv j \}$.
\end{definition}
It follows that $i \equiv j$ if and only if $i$ is connected
to $j$ by an alternating red/green path or cycle. The
equivalence classes $X$ with respect to $\equiv$ are maximal
length paths and cycles, as depicted in Figure~\ref{fig:redBluePathsCycles}.
Moreover, it is easy to see that elements of $A$ either belong to cycles or
are found at {\em interior} nodes of paths, while elements of $B$ are found 
exclusively at the left or right terminal nodes of paths.

Note that odd-length cycles cannot exist, due to the fact that a structure
cannot contain base triples -- see condition 2 of
Definition~\ref{def:secStr}. Moreover,  even-length cycles can indeed exist
-- consider, for instance, the structure 
$s$, whose only base pairs are $(1,15)$ and $(5,10)$, and the structure
$t$, whose only base pairs are  $(1,5)$ and $(10,15)$.
Then we have the red/green cycle
$1 \rightarrow 5 \rightarrow 10 \rightarrow 15 \rightarrow 1$,
consisting of red edge $1 \rightarrow 5$, since $(1,5) \in t$,
green edge $5 \rightarrow 10$, since $(5,10) \in s$,
red edge $10 \rightarrow 15$, since $(10,15) \in t$, and
green edge $15 \rightarrow 1$, since $(1,15) \in s$.

From the discussion before Definition~\ref{def:equivRelation}, it follows
that $A$ in equation~(\ref{eqn:defA}) consists of the 
nodes of every cycle together with all {\em interior} (yellow) nodes
of paths of type 1-4.  Moreover, we can think of $B$ in 
equation~(\ref{eqn:defB}) as consisting of all path {\em end nodes}, 
i.e. those that have only one incident edge.  Let
$B_1 \subseteq B$ [resp. $B_2 \subseteq B$] denote the set of elements of
$B$ that belong to type 1 paths [resp. type 4 paths] of length 1, i.e. positions
incident to isolated green [resp. red] edges that correspond to
base pairs $(i,j) \in s$ where $i,j$ are {\em not} touched by $t$
[resp. $(i,j) \in t$ where $i,j$ are {\em not} touched by $s$]. 
Let $B_0 = B-B_1-B_2$ be the set of end nodes of a path of length
2 or more.  Letting $BP_1$ [resp. $BP_2$] denote the set of base pairs $(i,j)$
that belong to $s$ and are not touched by $t$
[resp. belong to $t$ and are not touched by $s$], 
we can formalize the previous definitions as follows.

\begin{align}
\label{eqn:defB1}
B_1 &= \{ i \in [1,n]: \exists j \left[ \{ i,j \} \in s,
t(i)=0=t(j) \right] \\
\label{eqn:defB2}
B_2 &= \{ i \in [1,n]: \exists j \left[ \{ i,j \} \in t,
s(i)=0=s(j) \right] \\
\label{eqn:defB0}
B_0 &= B - (B_1 \cup B_2)  \\
\label{eqn:defBP1}
BP_1 &= \{ (i,j) \in s: t[i]=0=t[j] \}\\
\label{eqn:defBP2}
BP_2 &= \{ (i,j) \in t: s[i]=0=s[j] \}
\end{align}

In Appendix~\ref{section:pathsOneToFive}, it is proved that 
pk-$MS_2$ distance between 
$s \upharpoonright{X}$ and $t \upharpoonright{X}$
for any maximal path $X$ is equal to Hamming distance
$\lfloor \frac{d_H(s \upharpoonright{X},t \upharpoonright{X})}{2}
\rfloor$; in contrast,
pk-$MS_2$ distance between $s \upharpoonright{X}$ and $t \upharpoonright{X}$
for any cycle $X$ is equal to 
$\lfloor \frac{d_H(s \upharpoonright{X},t \upharpoonright{X})}{2} \rfloor +1$.  
It follows that 
\begin{align}
\label{eqn:relationBetweenHammingAndPKMS2distance}
d_{pk-MS_2}(s,t) &= \lfloor \frac{d_H(s,t)}{2} \rfloor
\end{align}
if and only if there are no type 5 paths (i.e. cycles).
This result justifies {\em ex post facto}
the use of Hamming distance in the investigation of RNA molecular evolution
\cite{schusterStadler:conformationalEvolution,Wagner.bj14}.
We also have the following.

\begin{lemma}
\label{lemma:pkMS2}
Let $s,t$ be two arbitrary (possibly pseudoknotted) 
structures for the RNA sequence
$a_1,\ldots,a_n$, and let $X_1,\ldots,X_m$ be the
equivalence classes with respect to equivalence relation $\equiv$ on
$A \cup B$.
Then the pk-$MS_2$ distance between $s$ and $t$ is equal to
\begin{align*}
d_{pk-MS_2}(s,t) &=\sum_{i=1}^m \max\Big( |s \upharpoonright{X_i}|,
|t \upharpoonright{X_i}| \Big)
\end{align*}
\end{lemma}

This lemma is useful, since the pk-$MS_2$ distance provides a lower bound 
for the $MS_2$ distance between any two secondary structures, and hence
allows a straightforward, but slow (exponential time)
branch-and-bound algorithm 
to be implemented for the exact $MS_2$ distance -- pseudocode for the
branch-and-bound algorithm is given in 
Section~\ref{section:branchAndBoundAlgorithm} of the Appendix.
To compute pk-$MS_2$ distance, we remove those base pairs in $s-t$ that
are not touched by $t$, compute the equivalence classes (connected components)
$X$ on the set of positions belonging to the remaining base pairs (provided that
the position does not belong to a common base pair of both $s$ and $t$),
then determine for each $X$ a minimum length pk-$MS_2$ folding pathway from
$s \upharpoonright{X}$ to $t \upharpoonright{X}$. The formal pseudocode
follows.

\begin{algorithmPeter}[pk-$MS_2$ distance]
\label{algo:ms2pathLengthArbitraryStr}
$MS_2$-path length between two possibly pseudoknotted structures $s,t$.
\end{algorithmPeter}
\bigskip
\begin{small}
\mverbatim
 1. remove from $s$ all base pairs of $BP_1$ 
 2. $numMoves = ||BP_1||$ 
 3. $Q = A \cup B_0$ 
 4. while $Q \ne \emptyset$ \{
 5.   $x_0 = \min(Q)$; $X=[x_0]$ //$X$ is equivalence class of $x_0$
 6.   determine path type of $X$
 7.   compute minimum length folding pathway from $s \upharpoonright{X}$ to $t \upharpoonright{X}$
 8.   $numMoves = numMoves + \max(||s\upharpoonright{X}||, ||t\upharpoonright{X}||)$
 9. add to $s$ all base pairs in $BP_2$ 
10. $numMoves = numMoves + ||BP_2||$
11. return $numMoves$
|mendverbatim
\end{small}

Straightforward details of how to implement line 7 are given in the
Appendix.  The principle underlying the reason that 
Algorithm~\ref{algo:ms2pathLengthArbitraryStr} produces a minimum
length (pseudoknotted) $MS_2$ folding trajectory from $s$ to $t$ is that
we {\em maximize} the number of shift moves, since a single shift move from
$\{ x,y \} \in s$ to $\{ y,z \} \in t$ corresponds to the simultaneous
removal of $\{ x,y\}$ and addition of $\{y,z\}$. We apply this principle
in the next section to determine the minimum 
length (non-pseudoknotted) $MS_2$ folding trajectory from $s$ to $t$.

\section{$MS_2$ distance between secondary structures}
\label{section:exactIPalgorithm}

In this section, we present an integer programming (IP) algorithm to 
compute the $MS_2$ distance between any two secondary structures $s,t$, i.e.
the minimum length of an $MS_2$ trajectory from $s$ to $t$.
Our algorithm has been cross-checked with the exhaustive
branch-and-bound algorithm mentioned at the end of the last section.

As in the previous section, our goal is to maximize the number of shift
operations in the $MS_2$ trajectory, formalized in the following simple
theorem, whose proof is clear.

\begin{theorem}
Suppose that the $MS_2$ distance between secondary structures $s,t$ is
$k$, i.e. base pair distance $d_{BP}(s,t) = |s-t| + |t-s| = k$.
Suppose that $\ell$ is the number of shift moves occurring
in a minimum length $MS_2$ refolding trajectory $s=s_0,s_1,\ldots,s_{m} = t$ 
from $s$ to $t$.
Then the $MS_2$ distance between $s$ and $t$ equals
\begin{align}
\label{thm:maxShiftMoves}
d_{MS_2}(s,t) = \ell + (k-2 \ell) = k-\ell
\end{align}
\end{theorem}
Our strategy will now be to use a graph-theoretic approach to maximize the
number of shift moves.

\subsection{RNA conflict digraph}
\label{subsection:RNAconflictDigraph}

Throughout this section, we take $s,t$ to be two arbitrary, distinct, but fixed
secondary structures of the RNA sequence $a_1,\ldots,a_n$.
Recall the definitions of $A,B,C,D$ in equations
(\ref{eqn:defA}--\ref{eqn:defD}), so that 
$A$ is the set of positions $x \in [1,n]$ that are base-paired in both
$s$ and $t$, but the base pairs in $s$ and $t$ are not identical; 
$B$ is the set of positions $x \in [1,n]$ that are base-paired in one of
$s$ or $t$, but not both;
$C$ is the set of positions $x \in [1,n]$ that are base-paired in neither
$s$ nor $t$, and 
$D$ is the set of positions $x \in [1,n]$ that are base-paired to the same
partner in both $s$ and $t$.  

To determine a minimum length $MS_2$ folding trajectory from secondary
structure $s$ to secondary structure $t$ is to maximize the number of
shift moves and minimize the number of base pair additions and removals. 
To that end, note that the
base pairs in $s$ that do not touch any base pair of $t$ must be removed
in any $MS_2$ path from $s$ to $t$, since there is no shift of such base 
pairs to a base pair of $t$ -- such base pairs are exactly those in
$BP_1$, defined in equation~(\ref{eqn:defBP1}).
Similarly, note that the
base pairs in $t$ that do not touch any base pair of $s$ must occur must be
added, in the transformation of $s$ to $t$, since there is no shift of any 
base pair from s to obtain such base pairs of $t$ --
such base pairs are exactly those in
$BP_2$, defined in equation~(\ref{eqn:defBP2}).
We now focus on the remaining base pairs
of $s$, all of which touch a base pair of $t$, and hence could theoretically
allow a shift move in transforming $s$ to $t$, {\em provided} that there
is no base triple or 
pseudoknot introduced by performing such a shift move. 
Examples of all six possible types of shift move are illustrated in
Figure~\ref{fig:sixShiftMovesRedGreen}.  To handle
such cases, we define the notion of {\em RNA conflict digraph}, solve
the {\em feedback vertex set} (FVS) problem
\cite{karpNPcomplete} by integer programming (IP), 
apply topological sorting \cite{cormen} 
to the acyclic digraph obtained by  removing a minimum set of vertices
occurring in feedback loops, then apply
shift moves in topologically sorted order. We now formalize this argument.

Define the digraph $G=(V,E)$, whose vertices
(or nodes) $n \in V$ are defined in the following
Definition~\ref{def:digraphNode} and whose directed edges are defined in 
Definition~\ref{def:digraphEdge}. 

\begin{definition} [Vertex in an RNA conflict digraph] \hfill\break
\label{def:digraphNode}
If $s,t$ are distinct secondary structures for the RNA sequence
$a_1,\ldots,a_n$, then a vertex in the RNA conflict digraph $G=G(s,t)$ is 
a triplet node, or more simply, node $v = (x,y,z)$ 
consisting of integers $x,y,z$, such that the base pair
$\{ x,y \}_< = (\min(x,y),\max(x,y))$ belongs to $t$, and the base pair
$\{ y,z \}_< = (\min(y,z),\max(y,z))$ belongs to $s$. Let $v.t$ [resp. 
$v.s$] denote the base pair $\{x,y\}_<$ [resp. $\{ y,z \}_<$] belonging
to $t$  [resp. $s$].  The middle integer $y$ of node $v=(x,y,z)$ is called
the {\em pivot} position, since it is common to both $s$ and $t$. Nodes
are ordered by the integer ordering of their pivot positions:
$(x,y,z) \preceq (x',y',z')$ if and only if $y \leq y'$
(or $y=y'$ and $x<x'$, or $y=y'$, $x=x'$, and $z<z'$).
If $v=(x,y,z)$ is a node, then $flatten(v)$ is
defined to be the set $\{ x,y,z\}$ of its coordinates.
\end{definition}
Nodes are representations of a potential shift move, and
can be categorized into six types, as shown in 
Figure~\ref{fig:sixShiftMovesRedGreen}.

\begin{definition}[Directed edge in an RNA conflict digraph] \hfill\break
\label{def:digraphEdge}
Base pairs $\{a,b\}_<$ and $\{c,d \}_<$ are said to touch if
$|\{a,b\} \cap \{c,d\}| = 1$;
in other words, base pairs touch if they form a base triple.
Base pairs $\{a,b\}_<$ and $\{ c,d\}_<$ are said to {\em cross} 
if either $\min(a,b) < \min(c,d) < \max(a,b) < \max(c,d)$ or
$\min(c,d) < \min(a,b) < \max(c,d) < \max(a,b)$;
in other words, base pairs cross if they form  a pseudoknot.
There is a directed edge from node
$n_1=(x_1,y_1,z_1)$ to node $n_2=(x_2,y_2,z_2)$, denoted by
$(n_1,n_2)\in E$ or equivalently by $n_1 \rightarrow n_2$, if
(1) $|flatten(n_1) \cap flatten(n_2)| \leq 1$, or in other words if
$n_1$ and $n_2$ overlap in at most one position, and
(2) the base pair $\{ y_1,z_1\}_< \in s$ from $n_1$ either touches or crosses
the base pair $\{ x_2,y_2\}_< \in t$ from $n_2$.
\end{definition}

Note that if the base pair $\{ y_1,z_1\}_< \in s$ from $n_1$ touches
the base pair $\{ x_2,y_2\}_< \in t$ from $n_2$, then it must be that
$z_1=x_2$; indeed, since each pivot node $y_1$ [resp. $y_2$] belongs to
a base pair of both $s$ and $t$, it cannot be that $z_1=y_2$ 
(because then $\{ y_1,z_1 \}_< \in s$ and $\{ y_2,z_2\}_< \in s$ would
form a base triple in $s$ at $z_1=y_2$), nor can it be that $y_1=x_2$
(because then $\{ x_1,y_1 \}_< \in t$ and $\{ x_2,y_2\}_< \in t$ would
form a base triple in $t$ at $y_1=x_2$).  Note as well that if
$n_1=(x_1,y_1,z_1)$ and $n_2=(x_2,y_2,z_2)$ are triplet nodes, then
$|flatten(n_1) \cap flatten(n_2)|=1$ implies that either
$n_1 \rightarrow n_2$ or $n_2 \rightarrow n_1$. 
Indeed, if there is a common element shared by
$n_1$ and $n_2$, then it cannot be a pivot element, since $s$ and $t$
cannot have a base triple. For the same reason, the common element cannot
belong to the base pairs $\{x_1,y_1\} \in t$ of $n_1$ and
$\{x_2,y_2\} \in t$ of $n_2$ (otherwise $t$ would contain a base triple),
nor can the common element
belong to the base pairs $\{y_1,z_1\} \in s$ of $n_1$ and
$\{y_2,z_2\} \in s$ of $n_2$ (otherwise $s$ would contain a base triple).
It follows that either $\{ x_1,y_1 \} \cap \{ y_2,z_2\} \ne \emptyset$,
or $\{ x_2,y_2 \} \cap \{ y_1,z_1\} \ne \emptyset$. From the assumption that
$|flatten(n_1) \cap flatten(n_2)|=1$, this implies that either
$n_2 \rightarrow n_1$ or that $n_1 \rightarrow n_2$, but not both.
Finally, note that if $n_1 = (x_1,y_1,z_1)$, $n_2 = (x_2,y_2,z_2)$ and
$|flatten(n_1) \cap flatten(n_2)|=2$, then there
are exactly three possibilities, all of which can be realized: 
\begin{enumerate}
\item
$n_1.t=n_2.t$, so that $\{ x_1,y_1 \} = \{ x_2,y_2 \}$, 
as in the example
$(1,5) \in s$, $(10,15)\in s$, $(5,10)\in t$, $n_1 = (10,5,1)$,
$n_2=(5,10,15)$;
\item
$n_1.s=n_2.s$, so that $\{ y_1,z_1 \} = \{ y_2,z_2 \}$,
as in the example
$(1,5) \in t$, $(10,15)\in t$, $(5,10)\in s$, $n_1 = (1,5,10)$,
$n_2=(15,10,5)$;
\item
$\{ x_1,z_1 \} = \{ x_2,z_2 \}$, as shown in 
Figure~\ref{fig:specialClosed2cycle}. This latter example will be called
a {\em closed 2-cycle}.
\end{enumerate}

These considerations produce the equivalent but sharper following definition.
\begin{definition}[Conflict digraph $G=(V,E)$]
\label{def:conflictDigraphG}
Let $s,t$ be distinct secondary structures for the RNA sequence 
$a_1,\ldots,a_n$. The RNA confict digraph 
$G(s,t)=(V(s,t),E(s,t))$, or $G=(V,E)$ when $s,t$ are clear from 
context, is defined by
\begin{align}
\label{eqn:conflictDigraphVertexSetV}
V &= \{ (x,y,z): x,y,z \in [1,n] \land \{x,y \} \in t \land \{y,z\} \in s \} \\
& \nonumber \\
E &= \Big\{ (n_1,n_2): 
n_1=(x_1,y_1,z_1)\in V \land n_2=(x_2,y_2,z_2) \in V \land  \nonumber \\
& |flatten(n_1) \cap flatten(n_2)| \leq 1 \land 
\Big( z_1=x_2 \lor \nonumber \\
& \Big( 
\left[
\min(y_1,z_1) < \min(x_2,y_2) < \max(y_1,z_1) < \max(x_2,y_2) \right] \lor
\nonumber \\ 
& \left[
\min(x_2,y_2) < \min(y_1,z_1) < \max(x_2,y_2) < \max(y_1,z_1) \right]
\Big) \Big\} 
\label{eqn:conflictDigraphEdgeSetE}
\end{align}
\end{definition}
The set of directed edges of conflict digraph $G=(V,E)$, as defined in
Definition~\ref{def:conflictDigraphG},
establishes a {\em partial ordering} on vertices of $V$ 
with the property that $n_1 \rightarrow n_2$ holds for vertices
$n_1 = (x,y,z)$, $n_2 = (u,v,w)$  if and only if 
(1) $n_1$ and $n_2$ overlap in at most one position, and
(2) when shift move
$n_2$ is applied, shifting $\{ v,w\} \in s$ to $\{ u,v\} \in t$,
the base pair $\{ u,v \}$ either touches or crosses the base pair
$\{ y,z \} \in s$ in $n_1$. It follows that
if $n_1 \rightarrow n_2$, then the shift move in which
$\{ y,z\} \in s$ shifts to $\{ x,y\} \in t$ {\em \underline{must}} be
performed {\em \underline{before}} the shift move where
$\{ v,w\} \in s$ shifts to $\{ u,v\} \in t$ -- indeed, if shifts are performed
in the opposite order, then after shifting $\{ v,w\} \in s$ to $\{u,v\} \in t$
and before shifting $\{ y,z\} \in s$ to $\{ x,y \} \in t$, we would create
either a base triple or a pseudoknot.  

Strictly speaking, the overlap condition (1) is not a necessary requirement,
and in our first IP algorithm to compute $MS_2$ distance \cite{wabi2017},
we considered a somewhat
more general edge relation without condition (1). If distinct vertices
$n_1,n_2$ violate condition (1), then
$|flatten(n_1)\cap flatten(n_2)| = 2$, and the constraint
($\ddag$) in line 7 of Algorithm~\ref{algo:MS2path} would ensure that at most 
one of $n_1,n_2$ are selected by the IP solver for the set $\Vbar$, in the 
resulting acyclic digraph $\Gbar = (\Vbar,\Ebar)$. 
Nevertheless, the
run time of our algorithm depends heavily on the number of simple directed
cycles in the initial conflict digraph $G=(V,E)$. Without condition (1),
nodes $n_1, n_2$ in a closed 2-cycle (see Figure~\ref{fig:specialClosed2cycle})
satisfy $n_1 \rightarrow n_2$ and $n_2 \rightarrow n_1$. Since it is possible
that $n_1$ belong to other cycles that do not contain $n_2$, this can
(and does) lead to a greatly increased number of directed cycles, hence
much longer run time -- indeed, our algorithm in \cite{wabi2017} runs
10 times slower than the current algorithm.

Before proceeding with the description of the algorithm, we must explain
how to treat {\em closed 2-cycles}, as shown in
Figure~\ref{fig:specialClosed2cycle}, for which there exist four integers 
$a_1<a_2<a_3<a_4$,
such that either Case A or Case B holds.
\smallskip

\noindent
{\sc \bf Case A:}
Base pairs $(a_1,a_2)$ and $(a_3,a_4)$ belong to $t$,
while base pairs $(a_1,a_4)$ and $(a_2,a_3)$ belong to $s$, 
as shown in Figure~\ref{fig:specialClosed2cycle}a.

\noindent
In this case, the conflict digraph $G=(V,E)$ contains the following 4 vertices
$v_1 = (a_1,a_2,a_3)$ of type 1, $v_2 = (a_3,a_4,a_1)$ of type 5,
$v_3 = (a_2,a_1,a_4)$ of type 4, and $v_4 = (a_4,a_3,a_2)$ of type 2. 
The overlap of any two distinct vertices has size 2, so 
by Definition~\ref{def:conflictDigraphG}, there can be no directed edge
between any vertices.  There are four optimal trajectories of size 3; for
specificity we select the following optimal trajectory:
\begin{align}
\label{eqn:trajectoryClosed2cycleA}
&\mbox{remove $(a_1,a_4)$ from $s$}
\end{align}
\smallskip

\noindent
{\sc \bf Case B:}
Base pairs $(a_1,a_2)$ and $(a_3,a_4)$ belong to $s$,
while base pairs $(a_1,a_4)$ and $(a_2,a_3)$ belong to $t$,
as shown in Figure~\ref{fig:specialClosed2cycle}b.

\noindent
In this case, the conflict digraph $G=(V,E)$ contains the following 4 vertices
$v_1 = (a_1,a_4,a_3)$  of type 6,
$v_2 = (a_4,a_1,a_2)$  of type 3,
$v_3 = (a_2,a_3,a_4)$  of type 1, and
$v_4 = (a_3,a_2,a_1)$  of type 2.
The overlap of any two distinct vertices has size 2, so 
by Definition~\ref{def:conflictDigraphG}, there can be no directed edge
between any vertices.  There are four optimal trajectories of size 3; for
specificity we select the following optimal trajectory:
\begin{align}
\label{eqn:trajectoryClosed2cycleB}
&\mbox{remove $(a_1,a_2)$ from $s$}
\end{align}

In Algorithm~\ref{algo:MS2path} below, it is necessary to 
list all closed 2-cycles, as depicted in 
Figure~\ref{fig:specialClosed2cycle}. This can be done in linear time
$O(n)$, for RNA sequence ${\bf a}=a_1,\ldots,a_n$ and secondary structures
$s,t$ by computing equivalence classes as 
defined in Definition~\ref{def:equivRelation}, then inspecting
all size 4 equivalence classes $X = \{ a_1,a_2,a_3,a_4 \}$ to determine
whether Case A or Case B applies. For each such closed 2-cycle,
Algorithm~\ref{algo:MS2path} computes the partial trajectory
(\ref{eqn:trajectoryClosed2cycleA}) or
(\ref{eqn:trajectoryClosed2cycleB}) appropriately, then the vertices
$v_1,v_2,v_3,v_4$ are deleted. No edges need to be deleted, since there
are no edges between $v_i$ and $v_j$ for $1 \leq i,j \leq 4$.
In creating the partial trajectories, the variable {\tt numMoves} must be
updated.

By special treatment of closed 2-cycles, we obtain a
10-fold speed-up in the exact IP Algorithm \ref{algo:MS2path}
over that of the precursor Algorithm 10 in our
WABI 2017 proceedings paper \cite{wabi2017} -- compare run times
of  Figure~\ref{fig:runTimeBenchmark} with those from Figure 5 of
\cite{wabi2017}.
Except for the special case of closed 2-cycles that must be handled before
general treatment, note that Definition~\ref{def:conflictDigraphG}
establishes a partial ordering on vertices of the conflict digraph
$G=(V,E)$, in that edges determine the order in which
shift moves should be performed. Indeed, if $n_1 = \{ x,y,z \}$, 
$n_2 = \{ u,v,z \}$  and $(n_1,n_2)\in E$, which we denote from now
on by $n_1 \rightarrow n_2$, then the shift move in which
$\{ y,z\} \in s$ shifts to $\{ x,y\} \in t$ {\em \underline{must}} be
performed {\em \underline{before}} the shift move where
$\{ v,w\} \in s$ shifts to $\{ u,v\} \in t$ -- indeed, if shifts are performed
in the opposite order, then after shifting $\{ v,w\} \in s$ to $\{u,v\} \in t$
and before shifting $\{ y,z\} \in s$ to $\{ x,y \} \in t$, we would create
either a base triple or a pseudoknot. Our strategy to efficiently compute
the $MS_2$ distance between secondary structures $s$ and $t$ will be to
(1) enumerate all simple cycles in the conflict digraph $G=(V,E)$ and to
(2) apply an integer programming (IP) solver to solve the minimum feedback arc
set problem $V' \subset V$. Noticing that the {\em induced digraph}
$\Gbar=(\Vbar,\Ebar)$, where $\Vbar=V-V'$ and 
$\Ebar = E \cap (\Vbar \times \Vbar)$, is
acyclic, we then (3) topologically sort $\Gbar$, and (4) perform shift moves
from $\Vbar$ in topologically sorted order.

\begin{algorithmPeter}[$MS_2$ distance from $s$ to $t$] 
\label{algo:MS2path}
\hfill\break
{\sc Input:} Secondary structures $s,t$ for RNA sequence $a_1,\ldots,a_n$
\hfill\break
{\sc Output:} Folding trajectory
$s = s_0,s_1,\ldots,s_m = t$, where $s_0,\ldots,s_m$ are
secondary structures, $m$ is the minimum possible value for which
$s_{i}$ is obtained from $s_{i-1}$ by a single base pair addition, removal or
shift for each $i=1,\ldots,m$.
\end{algorithmPeter}
First, initialize the variable {\tt numMoves} to $0$, and the list
{\tt moveSequence} to the empty list {\tt [ ]}. Recall that
$BP_2 = \{ (x,y) : (x,y) \in t, (s-t)[x]=0, (s-t)[y]=0\}$.
Bear in mind that $s$ is constantly being updated, so actions performed
on $s$ depend on its current value.
\bigskip
\begin{small}
\mverbatim
  //remove base pairs from $s$ that are untouched by $t$
 1. $BP_1 = \{ (x,y) : (x,y) \in s, (t-s)[x]=0, (t-s)[y]=0\}$
 2. for $(x,y) \in BP_1$
 3.   remove $(x,y)$ from $s$; numMoves = numMoves+1
  //define conflict digraph $G=(V,E)$ on updated $s$ and unchanged $t$
 4. define $V$ by equation (\ref{eqn:conflictDigraphVertexSetV})
 5. define $E$ by equation (\ref{eqn:conflictDigraphEdgeSetE})
 6. define conflict digraph $G=(V,E)$ 
  //IP solution of minimum feedback arc set problem
 7. maximize $\sum_{v \in V} x_v$ where $x_v \in \{0,1\}$, subject to constraints ($\dag$) and ($\ddag$)
  //constraint to remove vertex from each simple cycle of $G$
  ($\dag$) $\sum\limits_{v \in C} x_v < ||C||$ for each simple directed cycle $C$ of $G$ 
  //constraint to ensure shift moves cannot be applied if they share same base pair from $s$ or $t$
  ($\ddag$) $x_{v} + x_{v'} \leq 1$, for all pairs of vertices $v=(x,y,z)$ and $v'=(x',y',z')$ with $||\{ x,y,z\} \cap \{ x',y',z'\}||=2$
  //define IP solution acyclic digraph $\Gbar = (\Vbar,\Ebar)$
 8. $\Vbar = \{ v \in V: x_v = 1\}$; $V' = \{ v \in V: x_v = 0\}$
 9. $\Ebar = \{ (v,v'):  v,v' \in \Vbar \land (v,v') \in E\}$
10. $\Gbar = (\Vbar,\Ebar)$
  //handle special, closed 2-cycles
11. for each closed 2-cycle $[x]=\{a_1,a_2,a_3,a_4\}$ as depicted in Figure \ref{fig:specialClosed2cycle}
12.   if $[x]$ is of type A as depicted in Figure \ref{fig:specialClosed2cycle}a
13.     remove base pair from $s$ by equation (\ref{eqn:trajectoryClosed2cycleA}) 
14.   if $[x]$ is of type B as depicted in Figure \ref{fig:specialClosed2cycle}b
15.     remove base pair from $s$ by equation (\ref{eqn:trajectoryClosed2cycleB}) 
  //remove base pairs from $s$ that are not involved in a shift move
16. $\Vbar.s = \{ (x,y): \exists v \in \Vbar ( v.s = (x,y) ) \}$
17. for $(x,y) \in s-t$
18.   if $(x,y) \not\in \Vbar.s$
19.     remove $(x,y)$ from $s$; numMoves = numMoves+1 
  //topological sort for IP solution $\Gbar = (\Vbar,\Ebar)$
20. topological sort of $\Gbar$ using DFS \cite{cormen} to obtain total ordering $\prec$ on $\Vbar$
21. for $v=(x,y,z) \in \Vbar$ in topologically sorted order $\prec$
22.   shift $\{y,z\}$ to $\{x,y\}$ in $s$; numMoves = numMoves+1
  //add remaining base pairs from $t-s$, e.g. from $BP_2$ and type 4,5 paths in Figure \ref{fig:redBluePathsCycles} 
23. for $(x,y)\in t-s$
24.   add $(x,y)$ to $s$; numMoves = numMoves+1
25. return folding trajectory, numMoves
|mendverbatim
\end{small}
\bigskip

We now illustrate the definitions and the execution of the algorithm 
for a tiny example where $s = \{ (1,5), (10,15), (20,25) \}$ and
$t=\{ (5,10), (15,20) \}$. From Definition~\ref{def:equivRelation},
there is only one equivalence class $X=\{ 1,5,10,15,20,25 \}$ and 
it is a path of type 1, as illustrated 
in Figure~\ref{fig:redBluePathsCycles}, where $b_1=1$, $a_1=5$,
$b2=10$, $a_2=15$, $b_3=20$, $a_3=25$. From
Definition~\ref{def:digraphNode}, there are 4 vertices in
the conflict digraph $G=(V,E)$, where $v_1 = (10,5,1)$, 
$v_2 = (5,10,15)$, $v_3 = (20,15,10)$, $v_4 = (15,20,25)$ --
recall the convention from that definition that vertex $v=(x,y,z)$
means that base pair $\{ y,z \} \in s$ and base pair $\{ x,y\} \in t$,
so that the pivot position $y$ is shared by base pairs from both 
$s$ and $t$. From
Definition~\ref{def:digraphEdge}, there are only two directed edges,
$v_1 \rightarrow v_3$ since $v_1.s$ touches $v_2.t$, and
$v_2 \rightarrow v_4$ since $v_2.s$ touches $v_4.t$. Note there is
no edge from $v_1$ to $v_2$, or from $v_2$ to $v_3$, or from
$v_3$ to $v_4$, since their overlap has size 2 -- for instance
$flatten(v_1)=\{ 1,5,10\}$, $flatten(v_2)=\{ 5,10,15\}$, and
$flatten(v_1) \cap flatten(v_2)= \{ 5,10\}$ of size 2. There is
no cycle, so the constraint ($\dag$) in line 7 of
Algorithm~\ref{algo:MS2path} is not applied; however the constraint
($\ddag$) does apply, so that 
$x_{v_1}+x_{v_2} \leq 1$, $x_{v_2}+x_{v_3} \leq 1$,
$x_{v_3}+x_{v_4} \leq 1$. It follows that there are three possible
IP solutions for the vertex set $\Vbar$.
\smallskip

\noindent
{\sc Case 1:} $\Vbar = \{ v_1,v_3\}$ \hfill\break
Then $v_1.s = (1,5)$, $v_3.s = (10,15)$ so
$\Vbar.s =  \{ (1,5),(10,15)\}$ and by lines 11-14 we remove
base pair $(20,25)$ from $s$.
Now $\Gbar = (\Vbar,\Ebar)$, where $\Ebar=\{v_1 \rightarrow v_3\}$, so 
topological sort is trivial and we complete the trajectory by 
applying shift $v_1$ and then shift $v_3$. Trajectory length is 5.
\smallskip

\noindent
{\sc Case 2:} $\Vbar = \{ v_1,v_4\}$ \hfill\break
Then $v_1.s = (1,5)$, $v_4.s = (20,25)$ so
$\Vbar.s =  \{ (1,5),(20,25)\}$ and by lines 11-14 we remove
base pair $(10,15)$ from $s$.
Now $\Gbar = (\Vbar,\Ebar)$, where $\Ebar = \emptyset$, so
topological sort is trivial and we complete the trajectory by 
applying shift $v_1$ and then shift $v_4$, or by applying
shift $v_4$ and then shift $v_1$. Trajectory length is 5.
\smallskip

\noindent
{\sc Case 3:} $\Vbar = \{ v_2,v_4\}$ \hfill\break
Then $v_2.s = (10,15)$, $v_4.s = (20,25)$ so
$\Vbar.s =  \{ (10,15),(20,25)\}$ and by lines 11-14 we remove
base pair $(1,5)$ from $s$.
Now $\Gbar = (\Vbar,\Ebar)$, where $\Ebar = \{ v_2 \rightarrow v_4 \}$,
so topological sort is trivial and we complete the trajectory by 
applying shift $v_2$ and then shift $v_4$.  Trajectory length is 5.

\subsection{Examples to illustrate Algorithm~\ref{algo:MS2path}}
\label{section:traceOfMainAlgorithm}

We illustrate concepts defined so far with three examples: a toy 20 nt
RNA sequence, a 25 nt bistable switch, and the 56 nt spliced leader RNA from 
{\em L. collosoma}.

\subsubsection{Toy 20 nt sequence}

For the toy 20 nt sequence GGGAAAUUUC CCCAAAGGGG with initial structure $s$
shown in Figure~\ref{fig:cycleExample}a, and target structure $t$
shown in Figure~\ref{fig:cycleExample}b, the corresponding conflict
digraph is shown in Figure~\ref{fig:cycleExample}c.
This is a toy example, since the 
empty structure is energetically more favorable than either
structure: free energy of $s$ is +0.70 kcal/mol, while that for $t$ is
+3.30 kcal/mol. The conflict digraph contains 6 vertices, 10 directed
edges, and 3 simple cycles:
a first cycle
$\{(8, 20, 10), (9, 19, 11), (18, 10, 20), (19, 9, 1)\}$ of size 4,
a second cycle $\{[(8, 20, 10), (19, 9, 1)\}$ of size 2, and a third
cycle $\{(18, 10, 20), (9, 19, 11)\}$ of size 2.

\subsubsection{Bistable switch}

Figure~\ref{fig:conflictDigraphMicuraSwitch} depicts the secondary
structure for the metastable and the MFE structures, as well as the
corresponding conflict digraphs for the 25 nt bistable switch, with sequence
UGUACCGGAA GGUGCGAAUC UUCCG, taken from 
Figure 1(b).1 of \cite{Hobartner.jmb03}, in which the authors report
structural probing by comparative imino proton NMR spectroscopy.
The minimum free energy (MFE) structure has -10.20 kcal/mol, while
the next metastable structure has -7.40 kcal/mol.
Two lower energy structures exist, having -9.00 kcal/mol
resp. -7.60 kcal/mol; however, each is a minor variant of the MFE structure.
Figures~\ref{fig:conflictDigraphMicuraSwitch}a and
\ref{fig:conflictDigraphMicuraSwitch}b depict respectively the metastable 
and the MFE secondary structures for this 25 nt RNA, while 
Figures~\ref{fig:conflictDigraphMicuraSwitch}c and
\ref{fig:conflictDigraphMicuraSwitch}d depict respectively the
MFE conflict digraph and the metastable conflict digraph. 

For this 25 nt bistable switch, let $s$ denote the metastable structure and
$t$ denote the MFE structure. We determine the following.
, then we have the following.
\begin{align*}
s &= [(1, 16), (2, 15), (3, 14), (4, 13), (5, 12), (6, 11)] \mbox{ with
6 base pairs} \\ 
t &= [(6, 25), (7, 24), (8, 23), (9, 22), (10, 21), (11, 20), (12, 19), (13, 18)]  \mbox{ with 8 base pairs}\\
A &= \{6, 11, 12, 13\} \\
B &= \{1, 2, 3, 4, 5, 7, 8, 9, 10, 14, 15, 16, 18, 19, 20, 21, 22, 23, 24, 25\}\\
C &= \{ 17 \} \\
D &= \emptyset\\
BP_1 &= \{(1, 16), (2, 15), (3, 14) \} \mbox{ with 3 base pairs}\\
BP_2 &= \{ (7, 24), (8, 23), (9, 22), (10, 21) \} \mbox{ with 4 base pairs}\\
B_0 &= \{4, 5, 18, 19, 20, 25\}\\
B_1 &= \{1, 2, 3, 14, 15, 16 \}\\
B_2 &= \{7, 8, 9, 10, 21, 22, 23, 24 \}
\end{align*}
and there are three equivalence classes:
$X_1 = \{4, 13, 18\}$ of type 2, $X_2 = \{5, 12, 19\}$ of type 2, and
$X_3 = \{6, 11, 20, 25\}$ of type 4.

Figure~\ref{fig:conflictDigraphMicuraSwitch}c depicts the MFE
conflict digraph, where $s$ denotes the metastable structure and $t$ denotes  
the MFE structure. 
In the MFE conflict digraph $G=(V,E)$, vertices
are triplet nodes $(x,y,z)$,
where (unordered) base pair $\{ y,z\} \in s$ belongs to the
metastable [resp. MFE] structure, and (unordered) base pair $\{x,y\}\in t$
belongs to the MFE [resp. metastable] structure. A direct edge
$(x,y,z) \rightarrow (u,v,w)$ occurs if $\{ y,z\}\in s$ touches or crosses
$\{u,v\}\in t$.
Both the MFE and the metastable conflict digraphs are acyclic. Although
there are no cycles, the IP solver is nevertheless invoked in line 7 with
constraint ($\ddag$), resulting in {\em either} a first solution 
$\Vbar = \{ (18, 13, 4), (19, 12, 5), (20, 11, 6)\}$ or a second solution
$\Vbar = \{ (18, 13, 4), (19, 12, 5), (25, 6, 11)\}$. Indeed, the
overlap of vertices $(20, 11, 6)$ and $(25, 6, 11)$ has size 2, so one of
these vertices must be excluded from $\Vbar$ in 8 of 
Algorithm~\ref{algo:MS2path}. Assume that the first solution is returned
by the IP solver. Then we obtain the following minimum length $MS_2$
folding trajectory from metastable $s$ to MFE $t$.

Vertex and edge set of $G=(V,E)$ are given by the following.
\begin{align*}
V &= \{(18, 13, 4), (19, 12, 5), (20, 11, 6), (25, 6, 11)\}\\
E &= \{(18, 13, 4) \rightarrow (19, 12, 5), (18, 13, 4)
 \rightarrow (20, 11, 6), (18, 13, 4)  \rightarrow (25, 6, 11), \\
&(19, 12, 5)  \rightarrow (20, 11, 6), (19, 12, 5)  \rightarrow (25, 6, 11)\}
\end{align*}
One of two minimum length $MS_2$ folding trajectories is given by the following.
\bigskip

\begin{small}
\mverbatim
 1. UGUACCGGAAGGUGCGAAUCUUCCG
 2. 1234567890123456789012345

 0. ((((((....)))))).........   metastable $s$
 1. .(((((....)))))..........   remove (1,16) 
 2. ..((((....))))...........   remove (2,15) 
 3. ...(((....)))............   remove (3,14) 
 4. ....((....))(....).......   shift (4,13) to (13,18) 
 5. .....(....)((....))......   shift (5,12) to (12,19) 
 6. ..........(((....))).....   shift (6,11) to (11,20) 
 7. ......(...(((....)))...).   add (7,24)  
 8. ......((..(((....)))..)).   add (8,23)  
 9. ......(((.(((....))).))).   add (9,22)
10. ......(((((((....))))))).   add (10,21) 
11. .....((((((((....))))))))   add (6,25)
|mendverbatim
\end{small}
\bigskip

Algorithm~\ref{algo:MS2path} executes the following steps:
(1) Remove base pairs in $BP_1$ from $s$.
(2) Compute conflict digraph $G=(V,E)$.
(3) Apply IP solver to determine maximum size $\Vbar \subseteq V$,
subject to removing a vertex from each cycle ($\dag$) and not allowing
any two vertices in $\Vbar$ to have overlap of size 2.
(4) Topologically sorting the induced digraph $\Gbar=(\Vbar,\Ebar)$.
(5) Execute shifts according to total ordering $\prec$ given by topological
sort.
(6) Add remaining base pairs from $t-s$.
Note that in trajectory steps 7-10, the base pair added comes from $BP_2$,
while that in step 11 is a base pair from $t$ that is ``leftover'', due
to the fact that triplet node (shift move)
$(25,6,11)$ does not belong to IP solution $\Vbar$.

\subsubsection{Spliced leader from {\em L. collosoma}}

For the 56 nt {\em L. collosoma} spliced leader RNA, whose switching
properties were investigated in \cite{lecuyerCrothers} by
stopped-flow rapid-mixing and temperature-jump measurements,
the MFE and metastable structures are shown in Figure~\ref{fig:lecuyer},
along with the conflict digraph for $MS_2$ folding from the metastable
structure to the MFE structure.  This RNA
has sequence AACUAAAACA AUUUUUGAAG AACAGUUUCU GUACUUCAUU GGUAUGUAGA
GACUUC, an MFE structure having -9.40 kcal/mol, and an alternate metastable
structure having -9.20 kcal/mol.
Figure~\ref{fig:lecuyer} displays the MFE and metastable structures for
{\em L. collosoma} spliced leader RNA, along with the conflict digraph
for $MS_2$ folding from the metastable to the MFE structure. 

For {\em L. collosoma} spliced leader RNA, if we let $s$ denote the metastable
structure and $t$ denote the MFE structure, then there are seven equivalence
classes:
$X_1 = \{10, 45, 31, 23\}$ of type 4;
$X_2 = \{11, 43, 33\}$ of type 3;
$X_3 = \{12, 42, 34, 20\}$ of type 4,
$X_4 = \{13, 41, 35, 19\}$ of type 4,
$X_5 = \{22, 32, 44\}$ of type 3,
$X_6 = \{24, 54, 30, 48\}$ of type 1, and
$X_7 = \{25, 53, 29, 49\}$ of type 1.
As in the case with the 25 nt bistable switch,
the equivalence classes for the situation where $s$ and $t$ are interchanged
are {\em identical}, although type 1 paths become type 4 paths (and vice versa),
and type 2 paths become type 3 paths (and vice versa). Output from our (optimal) IP algorithm is as follows.
\bigskip

\begin{small}
\mverbatim
AACUAAAACAAUUUUUGAAGAACAGUUUCUGUACUUCAUUGGUAUGUAGAGACUUC
12345678901234567890123456789012345678901234567890123456

Number of Nodes: 12
Number of edges: 71
Number of cycles: 5
s: .......................((((((((((((.....)))))..)))))))..  -9.20 kcal/mol
t: .......((((((..(((((.((((...)))).)))))..))).))).........  -9.40 kcal/mol

 0. .......................((((((((((((.....)))))..)))))))..	metastable $s$
 1. .......................((.(((((((((.....)))))..)))).))..	remove	(26,52)
 2. .......................((..((((((((.....)))))..)))..))..	remove	(27,51)
 3. .......................((...(((((((.....)))))..))...))..	remove	(28,50)
 4. .......................((....((((((.....)))))..)....))..	remove	(29,49)
 5. .......................((.....(((((.....))))).......))..	remove	(30,48)
 6. .......................((...).(((((.....)))))........)..	(25,53)	-> (25,29)
 7. .......................((...))(((((.....)))))...........	(24,54)	-> (24,30)
 8. .........(.............((...)).((((.....)))))...........	(31,45)	-> (10,45)
 9. .........(...........(.((...)).)(((.....))).)...........	(32,44)	-> (22,32)
10. .........((..........(.((...)).).((.....))).)...........	(33,43)	-> (11,43)
11. .........(((.........(.((...)).)..(.....))).)...........	(34,42)	-> (12,42)
12. .........(((......(..(.((...)).)..)......)).)...........	(35,41)	-> (19,35)
13. .......(.(((......(..(.((...)).)..)......)).).).........	add	(8,47)
14. .......(((((......(..(.((...)).)..)......)).))).........	add	(9,46)
15. .......(((((...(..(..(.((...)).)..)..)...)).))).........	add	(16,38)
16. .......(((((...((.(..(.((...)).)..).))...)).))).........	add	(17,37)
17. .......(((((...((((..(.((...)).)..))))...)).))).........	add	(18,36)
18. .......(((((...(((((.(.((...)).).)))))...)).))).........	add	(20,34)
19  .......(((((...(((((.((((...)))).)))))...)).))).........	add	(23,31)
20. .......((((((..(((((.((((...)))).)))))..))).))).........	add	(13,41)

Number of base pair removals: 5
Number of base pair additions: 8
Number of base pair shifts: 7
MS2 Distance: 20
|mendverbatim
\end{small}
\bigskip

Figure~\ref{fig:lecuyer}a depicts the initial structure $s$, and
Figure~\ref{fig:lecuyer}b depicts the target minimum free energy structure
$t$ for spliced leader RNA from {\em L. collosoma}. The conflict digraph
for the refolding from $s$ to $t$ is shown in
Figure~\ref{fig:lecuyer}c.
Figure~\ref{fig:LcollosomaBis}a displays the 
rainbow diagram for spliced leader RNA from {\em L. collosoma}, 
in which the base pairs for the
initial structure $s$ (Figure~\ref{fig:lecuyer}a)
are shown below the line in red, while those for the
target structure $t$ (Figure~\ref{fig:lecuyer}b)
are shown above the line in blue.
Figure~\ref{fig:LcollosomaBis}c displays the Arrhenius tree,
where leaf index 2 represents the initial metastable structure $s$ with free 
energy -9.20 kcal/mol as shown in 
Figure~\ref{fig:lecuyer}a, while
leaf index 1 represents the target MFE structure $t$ with free 
energy -9.40 kcal/mol as shown in Figure~\ref{fig:lecuyer}b.
In Figure~\ref{fig:LcollosomaBis}b,
the dotted blue line depicts the free energies of structures in the shortest 
$MS_2$ folding trajectory for spliced leader, as computed by 
Algorithm~\ref{algo:MS2path}, while the solid red line depicts the 
free energies of the energy-optimal folding trajectory as computed by
the programs {\tt RNAsubopt}
\cite{wuchtyFontanaHofackerSchuster} and {\tt barriers} \cite{flammHofacker}.

\subsubsection{xpt riboswitch from {\em B. subtilis}}

In this section, we describe the shortest $MS_2$ folding trajectory from
the initial gene ON structure $s$ 
to the target gene OFF structure $t$ 
for the 156 nt xanthine phosphoribosyltransferase (xpt) riboswitch from 
{\em B. subtilis}, where the sequence and secondary structures are taken from
Figure 1A of \cite{breaker:Riboswitch2}. The gene ON [resp. OFF] structures
for the 156 nt xpt RNA sequence
AGGAACACUC AUAUAAUCGC GUGGAUAUGG CACGCAAGUU UCUACCGGGC ACCGUAAAUG
UCCGACUAUG GGUGAGCAAU GGAACCGCAC GUGUACGGUU UUUUGUGAUA UCAGCAUUGC
UUGCUCUUUA UUUGAGCGGG CAAUGCUUUU UUUAUU
are displayed in 
Figure~\ref{fig:conflictDigraphXPT}a 
[resp.  \ref{fig:conflictDigraphXPT}b], while 
Figure~\ref{fig:conflictDigraphXPT}c shows the {\em rainbow} diagram, where
lower red arcs [resp. upper blue arcs] indicate the base pairs of the
initial gene ON [resp. target gene OFF] structure. The default structure
for the xpt riboswitch in {\em B. subtilis} is the gene ON structure;
however, the binding of a guanine nucleoside ligand to cytidine in position 66 
triggers a conformational change to the gene OFF structure.
Figure~\ref{fig:conflictDigraphXPT}d depicts the conflict digraph 
$G=(V,E)$ containing 18 vertices, 113 directed
edges, and 1806 directed cycles, which is used to
compute the shortest $MS_2$ folding trajectory from the gene ON to the
gene OFF structure.
Figures~\ref{fig:xptStructuresAndTrajectory}a and 
\ref{fig:xptStructuresAndTrajectory}b show an enlargement 
of the initial gene ON structure $s$ and target gene OFF structure $t$, which
allows us to follow the moves in a shortest $MS_2$ trajectory
that is displayed in  Figure~\ref{fig:xptStructuresAndTrajectory}c.

\section{An algorithm for near-optimal $MS_2$ distance}
\label{section:nearOptimalAlgorithm}

Since the exact IP Algorithm~\ref{algo:MS2path} could not compute
the shortest $MS_2$ folding trajectories between the minimum free energy
(MFE) structure and Zuker suboptimal structures for some Rfam sequences
of even modest size ($\approx 100$ nt), we designed a near-optimal IP
algorithm (presented in this section), and a greedy algorithm
(presented in Section~\ref{subsection:greedyAlgorithm} of the Appendix).
The exact branch-and-bound algorithm from 
Appendix~\ref{section:branchAndBoundAlgorithm} was used to debug and
cross-check all algorithms.

The run time complexity of both the exact IP Algorithm~\ref{algo:MS2path} and
the greedy algorithm is due to the possibly exponentially large set of
directed simple cycles in the RNA conflict digraph.  By designing
a 2-step process, in which the {\em feedback arc set} (FAS) problem is first
solved for a coarse-grained digraph defined below, and subsequently
the {\em feedback vertex set} (FVS) problem is solved for each equivalence
class, we obtain a much faster algorithm to compute a near-optimal  
$MS_2$ folding  trajectory between secondary structures $s$ and $t$ 
for the RNA sequence $\{ a_1,\ldots,a_n \}$.  In the first step, we use 
IP to solve the {\em feedback arc set} (FAS) problem for a particular 
coarse-grained digraph defined below, whose vertices
are the equivalence classes as defined in Definition~\ref{def:equivRelation}.
The number of cycles for this coarse-grained digraph is quite manageable,
even for large RNAs, hence the FAS can be efficiently solved.
After removal of an arc from each directed cycle, topological sorting is 
applied to determine a total ordering according to which each individual
equivalence class is processed. 
In the second step, Algorithm~\ref{algo:MS2path}  is applied to each
equivalence class in topologically sorted order, whereby the feedback
vertex set (FVS) problem is solved for the equivalence class under 
consideration.  In the remainder of this section, we fill in the details for
this overview, and then present pseudocode
for the near-optimal Algorithm~\ref{algo:nearOptimalMS2path}.

Given secondary structures $s$ and $t$ 
for the RNA sequence $\{ a_1,\ldots,a_n \}$, we partition
the set $[1,n]$ into disjoint sets $A,B,C,D$ as in 
Section~\ref{section:ms2distancePseudoknottedStr} by
following equations (\ref{eqn:defA}), (\ref{eqn:defB}), (\ref{eqn:defC}),
(\ref{eqn:defD}). The union $A \cup B$ is subsequently partitioned into 
the equivalence classes $X_1,\ldots,X_m$, defined 
in Definition~\ref{def:equivRelation}. Define the coarse-grain,
conflict digraph $G=(V,E)$, whose vertices are the 
indices of equivalence classes $X_1,\ldots,X_m$,
and whose directed edges $i \rightarrow j$ are defined 
if there exists a base pair $(x,y)\in s$, $x,y \in X_i$ which crosses 
a base pair $(u,v)\in t$, $u,v \in X_j$. Although there
may be many such base pairs $(x,y) \in s$ and $(u,v) \in t$, there is
only one edge between $i$ and $j$; i.e. $G$ is a directed graph, not
a directed multi-graph.
If $i \rightarrow j$ is an edge, then
we define 
$N_{i,j}$ to be the set of all base pairs $(u,v) \in s$,
$u,v \in X_i$ that cross some base pair $(u,v)\in t$, $u,v \in X_j$, and let
$n_{i,j}$ the number of base pairs in $N_{i,j}$.
Formally, given
equivalence classes $X_1,\ldots,X_m$, the coarse-grain, conflict
{\em digraph} $G=(V,E)$ is defined by

\begin{align}
\label{eqn:heuristicDigraphVertices}
V &= \{ 1,\ldots,m \}\\
\label{eqn:heuristicDigraphEdges}
E &= \Big\{ i \rightarrow j: 
\exists (x,y) \in s \exists (u,v) \in t \big[ 
x,y \in X_i \land u,v \in X_j \big] \Big \}
\end{align}
A directed edge from $i$ to $j$ may be denoted either by
$i \rightarrow j \in E$ or by $(i,j) \in E$. For each edge $i \rightarrow j$,
we formally define $N_{i,j}$ and $n_{i,j}$ by the following.
\begin{align}
\label{eqn:NijForCoarseGrainedDigraph}
N_{i,j} &=  \Big\{ (x,y) \in s: x,y \in X_i \land
\exists (u,v) \in t \big[ 
u,v \in X_j  \land \mbox{ $(x,y)$ crosses $(u,v)$ } \big] \Big \} \\
\label{eqn:nijForCoarseGrainedDigraph}
n_{i,j} &= \Big| N_{i,j} \Big|
\end{align}

We now solve the feedback arc set (FAS) problem, rather than the feedback 
vertex set (FVS) problem, for digraph $G$, by applying 
an IP solver to solve the following optimization problem: 

\begin{quote}
\begin{small}
\mverbatim
 1. maximize $\sum_{(i,j) \in E} n_{i,j} \cdot x_{i,j}$ subject to constraint ($\sharp$):
    ($\sharp$) $\sum\limits_{\stackrel{(i,j) \in E}{i,j \in C}} x_{i,j} < ||C||$
    for every directed cycle $C = \left( i_1,i_2,i_3,\ldots,i_{k-1},i_k \right)$
|mendverbatim
\end{small}
\end{quote}
This IP problem can be quickly solved, since there is usually only
a modest number
of directed cycles for the coarse-grained digraph. For each directed
edge or arc $i \rightarrow j$ that is to be removed from
a directed cycle, we remove {\em \underline{all}}
base pair $(x,y) \in $ from structure $s$ that cross some base pair 
$(u,v) \in t$ for which $u,v \in X_j$.
Removal of certain base pairs from $s$ can disconnect some
previous equivalence classes into two or more connected components,
hence equivalence classes $X'_1,\ldots,X'_{m'}$ must be recomputed for
the updated structure $s$ and  (unchanged) structure $t$. The conflict,
digraph $G'=(V',E')$ is then defined by equations
(\ref{eqn:heuristicDigraphVertices}) and (\ref{eqn:heuristicDigraphEdges})
for (updated) $s$ and (unchanged) $t$.  Since $G'$ is now
acyclic, it can be topologically sorted, which determines an ordering
$\sigma(1),\ldots,\sigma(m')$ for processing equivalence classes
$X'_1,\ldots,X'_{m'}$. To process an equivalence class $X'$, we restrict
the exact Algorithm~\ref{algo:MS2path} to each equivalence class. Indeed,
to process equivalence class $X'$, we define
a (local) conflict digraph $G(X') = (V(X'),E(X'))$ defined as follows.

\begin{align}
\label{eqn:defVofXprime}
V(X') &= \{ (x,y,z): x,y,z \in X'
\land \{x,y \} \in t \land \{y,z\} \in s\\
\label{eqn:defEofXprime}
E(X') &= \{ (x,y,z) \rightarrow (u,v,w): (x,y,z) \in V(X') \land
(u,v,w) \in V(X') \land \\
& \quad  
\mbox{ $\{ x,y\} $ touches or crosses $\{v,w\}$ }
\} \nonumber
\end{align}

\begin{algorithmPeter}[Near-optimal $MS_2$ distance from $s$ to $t$] 
\label{algo:nearOptimalMS2path}
\hfill\break
{\sc Input:} Secondary structures $s,t$ for RNA sequence $a_1,\ldots,a_n$
\hfill\break
{\sc Output:} $s = s_0,s_1,\ldots,s_m = t$, where $s_0,\ldots,s_m$ are
secondary structures, $m$ is a near-optimal value for which
$s_{i}$ is obtained from $s_{i-1}$ by a single base pair addition, removal or
shift for each $i=1,\ldots,m$.
\end{algorithmPeter}
First, initialize the variable {\tt numMoves} to $0$, and the list
{\tt moveSequence} to the empty list {\tt [~ ]}. Define
$BP_1 = \{ (x,y) : (x,y) \in t, (t-s)[x]=0, (t-s)[y]=0\}$; i.e.
$BP_1$ consists of those base pairs in $s$ which are not touched by any
base pair in $t$. Define
$BP_2 = \{ (x,y) : (x,y) \in t, (s-t)[x]=0, (s-t)[y]=0\}$; i.e.
$BP_2$ consists of those base pairs in $t$ which are not touched by any
base pair in $s$.

\bigskip
\begin{small}
\mverbatim
  //remove base pairs from $s$ that are untouched by $t$
 1. for $(x,y) \in BP_1$
 2.   $s = s - \{ (x,y) \}$
 3.   numMoves = numMoves + 1
  //define equivalence classes on updated $s,t$
 4. $[1,n] = A \cup B \cup C \cup D$
 5. determine equivalence classes $X_1,\ldots,X_m$ with union $A \cup B$
  //define conflict digraph $G=(V,E)$ on collection of equivalence classes
 6. define $V = \{ 1,\ldots,m\}$
 7. define $E = \{ (i,j): 1\leq i,j \leq m\}$ by equation (\ref{eqn:heuristicDigraphEdges})
 8. define coarse-grain, conflict digraph $G=(V,E)$ 
  //IP solution of feedback arc set problem (not feedback vertex set problem)
 9. maximize $\sum_{(i,j)  \in E} n_{i,j} \cdot x_{i,j}$ subject to constraint ($\sharp$):
  //remove arc from each simple cycle where$n_{i,j}$ defined in equation (\ref{eqn:nijForCoarseGrainedDigraph})
    ($\sharp$) $\sum\limits_{\stackrel{(i,j) \in E}{i \rightarrow j \in C}} x_{i,j} < ||C||$ 
    for every directed cycle $C = \left( i_1,i_2,\ldots,i_{k-1},i_k \right)$
10. $\widetilde{E} = \{ (i,j) : x_{i,j} = 0 \}$ //$\widetilde{E}$ is set of edges that must be removed
  //process the IP solution $\widetilde{E}$
11. for $(i,j) \in \widetilde{E}$
12.   for $(x,y) \in N_{i,j}$ //$N_{i,j}$ defined in Definition \ref{eqn:NijForCoarseGrainedDigraph}
13.     $s = s - \{ (x,y) \}$ //remove base pair from $s$ belonging to feedback arc 
14.     numMoves = numMoves + 1
  //determine equivalence classes ${X'}_1,\ldots,{X'}_{m'}$ for updated $s$ and (unchanged) $t$
15. $[1,n] = A' \cup B' \cup C' \cup D'$
16. define ${X'}_1,\ldots,{X'}_{m'}$ whose union is $(A' \cup B')$
17. define $V' = \{ 1,\ldots,m'\}$
18. define $E' = \{ (i,j): 1\leq i,j \leq m'\}$ by equation (\ref{eqn:heuristicDigraphEdges}) 
19. define $G'=(V',E')$ //note that $G'$ is an acyclic multigraph
20. let $\sigma \in S_{m'}$ be a topological sort of $V'$//$S_{m'}$ denotes set of all permutations on $[1,m']$ 
    //process shifts in $X_{\sigma(i)}$ in topological order by adapting part of Algorithm \ref{algo:MS2path} 
21. for $i=1$ to $m'$
22.   define $V(X_{\sigma(i)})$ by equation (\ref{eqn:defVofXprime})
23.   define $E(X_{\sigma(i)})$ by equation (\ref{eqn:defEofXprime})
24.   define $G(X_{\sigma(i)})= (V(X_{\sigma(i)}), E(X_{\sigma(i)}))$ 
      //IP solution of minimum feedback vertex set problem 
25.   maximize $\sum_{v \in V(X_{\sigma(i)})} x_v$ where $x_v \in \{0,1\}$, subject to constraints ($\dag_i$) and ($\ddag_i$) 
      //first constraint removes vertex from each simple cycle of $G(X_{\sigma(i)})$
        ($\dag_i$) $\sum\limits_{v \in C} x_v < ||C||$ for each simple directed cycle $C$ of $G(X_{\sigma(i)})$ 
      //ensure shift moves cannot be applied if they share same base pair from $s$ or $t$
        ($\ddag_i$) $x_{v} + x_{v'} \leq 1$, for distinct vertices $v=(x,y,z),v'=(x',y',z')$ with $||\{ x,y,z\} \cap \{ x',y',z'\}||=2$
      //define the induced, acyclic digraph $\Gbar(X_{\sigma(i)})$ 
26.   $\Vbar(X_{\sigma(i)}) = \{ v \in V(X_{\sigma(i)}): x_v = 1\}$
27.   $\Ebar(X_{\sigma(i)}) = \{ (v,v'):  v,v' \in \Vbar \land (v,v') \in E(X_{\sigma(i)})\}$
28.   let $\Gbar(X_{\sigma(i)}) = (\Vbar(X_{\sigma(i)}),\Ebar(X_{\sigma(i)}))$
  //handle special, closed 2-cycles
29.   for each closed 2-cycle $[x]=\{a_1,a_2,a_3,a_4\}$ as depicted in Figure \ref{fig:specialClosed2cycle}
30.     if $[x]$ is of type A as depicted in Figure \ref{fig:specialClosed2cycle}a
31.       remove base pair from $s$ by equation (\ref{eqn:trajectoryClosed2cycleA})
32.     if $[x]$ is of type B as depicted in Figure \ref{fig:specialClosed2cycle}b
33.       remove base pair from $s$ by equation (\ref{eqn:trajectoryClosed2cycleB})
      //remove base pairs from $s$ that are not involved in a shift move
34.   $\Vbar(X_{\sigma(i)}).s = \{ (x,y): \exists v \in \Vbar(X_{\sigma(i)})  ( v.s = (x,y) ) \}$
35.   for $(x,y) \in s-t$
36.     if $(x,y) \not\in \Vbar(X_{\sigma(i)}).s$
37.       remove $(x,y)$ from $s$; numMoves = numMoves+1
        //topological sort of IP solution $\Vbar(X_{\sigma(i)})$
38.   topological sort of $\Gbar(X_{\sigma(i)})$ using DFS to obtain total ordering $\prec$ on $\Vbar(X_{\sigma(i)})$
38.   for $v=(x,y,z) \in \Vbar$ in topologically sorted order $\prec$
        //check if shift would create a base triple, as in type 1,5 paths from Figure \ref{fig:redBluePathsCycles}
40.     if $s[x]=1$ //i.e. $\{u,x\} \in s$ for some $u\in [1,n]$
41.       remove $\{ u,x \}$ from $s$; numMoves = numMoves+1 
42.     shift $\{y,z\}$ to $\{x,y\}$ in $s$; numMoves = numMoves+1
        //remove any remaining base pairs from $s$ that have not been shifted
43.   for $(y,z)\in s-t$ which satisfy $y,z \in X_{\sigma(i)}$
44.     if $(x,y,z) \in V(X_{\sigma(i)})$ where $x = t[y]$
45.       shift base pair $\{ y,z \} \in s$ to $\{ x,y \} \in t$; numMoves = numMoves+1
46.     else //$\{y,z\}$ is a remaining base pair of $s$ but cannot be applied in a shift
47.       remove $(x,y)$ from $s$; numMoves = numMoves+1
  //add remaining base pairs from $t$ to $s$
48. for $(x,y) \in t-s$
49.   $s = s \cup \{ (x,y) \}$
50.   numMoves = numMoves + 1 
51. return moveSequence, numMoves
|mendverbatim
\end{small}
\bigskip

\section{Benchmarking results}
\label{section:benchmarking}

\subsection{Random sequences}
\label{section:benchmarkingRandom}

Given a random RNA sequence ${\bf a} = a_1,\ldots a_n$ of length $n$, 
we generate a list $L$ of all possible base pairs, then choose with uniform
probability a base pairs $(x,y)$ from $L$, add $(x,y)$ to the secondary
structure $s$ being constructed, then remove all base pairs $(x',y')$ from
$L$ that either touch or cross $(x,y)$, and repeat these last three steps
until we have constructed a secondary structure having the desired number
($n/5$) of base pairs. If the list $L$ is empty before completion of the
construction of secondary structure $s$, then reject $s$ and start over.
The following pseudocode describes how we generated the benchmarking data
set, where for each sequence length $n=10,15,20,\cdots,150$ nt, twenty-five 
random RNA sequences were generated of length $n$, with probability of $1/4$ 
for each nucleotide, in which twenty secondary structures $s,t$ were
uniformly randomly generated for each sequence so that 40\% of the nucleotides 
are base-paired.

\begin{small}
\begin{quote}
\mverbatim
 1. for $n$ = 10 to 150 with step size 10
 2.   for numSeq = 1 to 25
 3.     generate random RNA sequence ${\bf a}=a_1,\ldots,a_n$ of length $n$
 4.     generate $20$ random secondary structures of ${\bf a}$
 5.     for all ${20 \choose 2}=190$ pairs of structures $s,t$ of ${\bf a}$
 6.       compute optimal and near-optimal $MS_2$ folding trajectories from $s$ to $t$
|mendverbatim
\end{quote}
\end{small}
\smallskip

The number of computations per sequence length is thus $25 \cdot 190 = 4750$,
so the size of the benchmarking set is $15 \cdot 4750 = 71,250$. This 
benchmarking set is used in 
Figures~\ref{fig:benchmarkingAccuracy} -- \ref{fig:benchmarkingAvgNumCycles}.
Figure~\ref{fig:benchmarkingAccuracy} compares various distance measures
discussed in this paper: $MS_2$ distance computed by the optimal IP
Algorithm~\ref{algo:MS2path}, approximate $MS_2$ distance computed by 
the near-optimal Algorithm~\ref{algo:nearOptimalMS2path}, 
$pk-MS_2$ distance that allows pseudoknotted intermediate structures,
$MS_1$ distance, and Hamming distance divided by 2. Additionally, this
figure distinguishes the number of base pair additions/removals and
shifts in the $MS_2$ distance.
 
Figure~\ref{fig:benchmarkingCorrelation}a  shows the scatter plots and
Pearson correlation coefficients all pairs of the distance measures:
$MS_2$ distance, near-optimal $MS_2$ distance, 
$pk-MS_2$ distance, Hamming distance divided
by $2$,  $MS_1$ distance.  In contrast to 
Figure~\ref{fig:benchmarkingCorrelation}, the second panel 
Figure~\ref{fig:benchmarkingCorrelation}b shows the length-normalized values.
It is unclear why $MS_2$ distance has a slightly higher length-normalized
correlation with both Hamming distance divided by 2 and $MS_1$ distance,
than that with approximate $MS_2$ distance, as computed by 
Algorithm~\ref{algo:nearOptimalMS2path} -- despite the fact that the latter
algorithm approximates $MS_2$ distance much better than either Hamming distance
divided by 2 or $MS_1$ distance.
Figure~\ref{fig:runTimeBenchmark} shows that run-time of
Algorithms~\ref{algo:MS2path} and
\ref{algo:nearOptimalMS2path}, where the former is broken down into
time to generate the set of directed cycles and the time for the IP solver.
Note that there is a 10-fold speed-up in Algorithm~\ref{algo:MS2path}
from this paper, compared with a precursor of this algorithm that appeared
in the proceedings of the Workshop in Bioinformatics (WABI 2017).
Since Algorithm~\ref{algo:nearOptimalMS2path} applies
Algorithm~\ref{algo:MS2path} to each equivalence class, there is a
corresponding, but less striking speed-up in the near-optimal algorithm.
Since run-time depends heavily on the number of directed cycles in the
conflict digraphs, 
Figure~\ref{fig:benchmarkingSizeConflictDigraphAndCycleLengthDistribution}a 
shows the
size of vertex and edge sets of the conflict digraphs for the benchmarking
data, and 
Figure~\ref{fig:benchmarkingSizeConflictDigraphAndCycleLengthDistribution}b 
depicts the cycle length
distribution for benchmarking data of length 150; for different lengths,
there are similar distributions (data not shown).
Finally, 
Figure~\ref{fig:benchmarkingSizeConflictDigraphAndCycleLengthDistribution}c 
showns the (presumably)
exponential increase in the number of directed cycles, as a function of
sequence length. Since Algorithm~\ref{algo:nearOptimalMS2path} does not
compute the collection of all directed cycles (but only those for each
equivalence class), the run time of Algorithm~\ref{algo:nearOptimalMS2path} 
appears to be linear in sequence length, compared to the (presumably)
exponential run time of Algorithm~\ref{algo:MS2path}.

\subsection{Rfam sequences}
\label{section:benchmarkingRfam}

In this section, we use data from the Rfam 12.0 database
\cite{Nawrocki.nar15} for analogous computations as those from
the previous benchmarking section.  For each Rfam family having
average sequence length less than 100 nt, one sequence is randomly
selected, provided that the
base pair distance between its MFE structure and its Rfam 
consensus structure is a minimum.
For each such sequence ${\bf a}$, the target structure $t$
was taken to be the secondary structure having minimum free energy among
all structures of ${\bf a}$ that are compatible with the Rfam consensus
structure, as computed by {\tt RNAfold -C} \cite{Lorenz.amb11}
constrained with the consensus structure of ${\bf a}$.
The corresponding initial structure $s$ for sequence ${\bf a}$
was selected from a Zuker-suboptimal structure, obtained by
{\tt RNAsubopt -z} \cite{Lorenz.amb11}, with the property that
$|d_{\mbox{\small BP}}(s,t) - d_{\mbox{\small H}}(s,t)| < 0.2 \cdot
d_{\mbox{\small BP}}(s,t)$. Since we know from 
Figure~\ref{fig:runTimeBenchmark} that run time of the optimal
IP Algorithm~\ref{algo:MS2path} depends on the number of
cycles in the corresponding RNA conflict digraph, the last criterion is
likely to result in a less than  astronomical number of cycles.
The resulting dataset consisted of 1333 sequences, some of
whose lengths exceed 100 nt. Nevertheless, the number of cycles in the
RNA constraint digraph of 22 of the 1333 sequences exceeded 50 million
(an upper bound set for our program),
so all figures described in this section are based on 1311 sequences 
from Rfam.

Figure~\ref{fig:distanceMeasureBenchmarkRfam} depicts the moving 
averages in centered windows $[x-2,x+2]$ of the following
distance measures for the $1311$ sequences extracted from
Rfam 12.0 as described.  Distance measures include
(1) optimal $MS_2$-distance computed
by the exact IP (optimal) Algorithm~\ref{algo:MS2path} (where
the number of base pair additions ($+$) or removals ($-$) is indicated,
along with the number of shifts),
(2) near-optimal $MS_2$-distance computed
by near-optimal Algorithm~\ref{algo:nearOptimalMS2path},
(3) Hamming distance divided by $2$,
(4) $MS_1$ distance aka base pair distance,
(5) pseudoknotted $MS_2$ distance (pk-$MS_2$) computed
from Algorithm~\ref{algo:ms2pathLengthArbitraryStr},
(6) optimal local $MS_2$ with parameter $d=10$, and
(7) optimal local $MS_2$ with parameter $d=20$. 
The latter values were computed
by a variant of the exact IP Algorithm~\ref{algo:MS2path} with
{\em locality parameter} $d$, defined to allow
base pair shifts of the form
$(x,y) \rightarrow (x,z)$ or $(y,x) \rightarrow (z,x)$ only
when $|y-z| \leq d$.  This data suggests that Hamming distance over
$2$ ($d_{\mbox{\small H}}(s,t)/2$) closely approximates the 
distance computed by near-optimal
Algorithm~\ref{algo:nearOptimalMS2path}, while pk-$MS_2$ distance
($d_{\mbox{\small pk-$MS_2$}}(s,t)$)
is a better approximation to $MS_2$ distance than 
is Hamming distance over $2$.
Figure~\ref{fig:RfamDataCorrelation} presents scatter plots
and Pearson correlation values when comparing
various distance measures using the Rfam data.
Figure~\ref{fig:RfamDataCorrelation}a [resp.
Figure~\ref{fig:RfamDataCorrelation}b] 
presents Pearson correlation [resp.
{\em normalized} Pearson correlation] values computed, where by
{\em normalized}, we mean that for each of the $1311$ extracted Rfam
sequences ${\bf a}$ with corresponding initial structure $s$ 
and target structure $t$, the 
{\em length-normalized} distance measures $d(s,t)/|{\bf a}|$ are
correlated.
Figure~\ref{fig:runTimeBenchmarkRfam} depicts the moving average
run times as a function of sequence length, where for given value
$x$ the run times are averaged for sequences having length in
$[x-2,x+2]$. Finally, Figure~\ref{fig:numberRfamSequences}
depicts the number of sequences of various lengths used in the
Rfam benchmarking set of 1311 sequences.

\section{Conclusion}
\label{section:conclusion}

In this paper, we have introduced the first optimal and near-optimal
algorithms to compute the shortest RNA secondary structure folding 
trajectories in which each intermediate structure is obtained from its
predecessor by the addition, removal or shift of a base pair; i.e. 
the shortest $MS_2$ trajectories.
Since helix zippering and defect diffusion employ shift moves, one
might argue that it is better to include shift moves when physical 
modeling RNA folding, and indeed the RNA folding kinetics simulation 
program {\tt Kinfold} \cite{flamm} uses the $MS_2$ move set by default. 
Using the novel notion of RNA conflict directed graph, we describe
an optimal and near-optimal algorithm to compute the shortest $MS_2$
folding trajectory. Such trajectories pass through substantially higher
energy barriers than trajectories produced by {\tt Kinfold}, which
uses Gillespie's algorithm \cite{gillespieStochasticSimulation1}
(a version of event-driven Monte Carlo simulation) to generate
physically realistic $MS_2$ folding trajectories. We have shown in
Theorem~\ref{thm:NPhardnessMS2} that
it is NP-hard to compute the $MS_2$ folding trajectory having minimum
energy barrier, and have presented
anecdotal evidence that suggests that it may also NP-hard to compute 
the shortest $MS_2$ folding trajectory.  For this reason, and because
of the exponentially increasing number of cycles 
(see Figure~\ref{fig:benchmarkingAvgNumCycles})
and subsequent time requirements of our optimal IP 
Algorithm~\ref{algo:MS2path}, it is unlikely that (exact) $MS_2$ distance
prove to be of much use in molecular evolution studies such as
\cite{Borenstein.pnas06,Wagner:robustness,GarciaMartin.bb16}. 
Nevertheless, Figures~\ref{fig:benchmarkingAccuracy} and
\ref{fig:distanceMeasureBenchmarkRfam} suggest that either
pk-$MS_2$ distance and/or near-optimal $MS_2$ distance may be a 
better approximation to (exact) $MS_2$ distance than using
Hamming distance, as done in
\cite{schusterStadler:conformationalEvolution,Wagner.bj14}.
However, given the high correlations between these measures, it 
is unlikely to make much difference in molecular evolution studies.

Our graph-theoretic formulation involving RNA conflict digraphs
raises some interesting mathematical questions partially addressed 
in this paper; in particular,
it would be very interesting to characterize the class of digraphs that
can be represented by RNA conflict digraphs, and to determine whether
computing the shortest $MS_2$ folding trajectory is $NP$-hard. We
hope that the results presented in this paper may lead to resolution of
these questions.

\bibliographystyle{plain}

\hfill\break\clearpage

\input{figures}

\hfill\break\clearpage
\appendix
\input{appendix}

\end{document}

%% file: figures.tex
\section{Figures}

\begin{figure}
\centering
\includegraphics[width=0.7\textwidth]{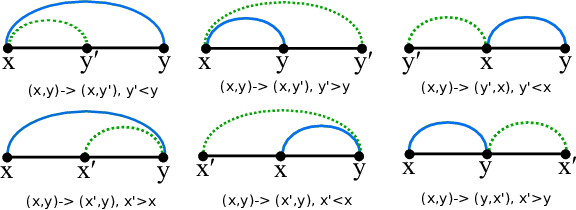}
\caption{Illustration of shift moves, taken from \cite{Clote.po15}.
}
\label{fig:shiftMoves}
\end{figure}

\begin{figure}
\centering
\includegraphics[width=0.8\textwidth]{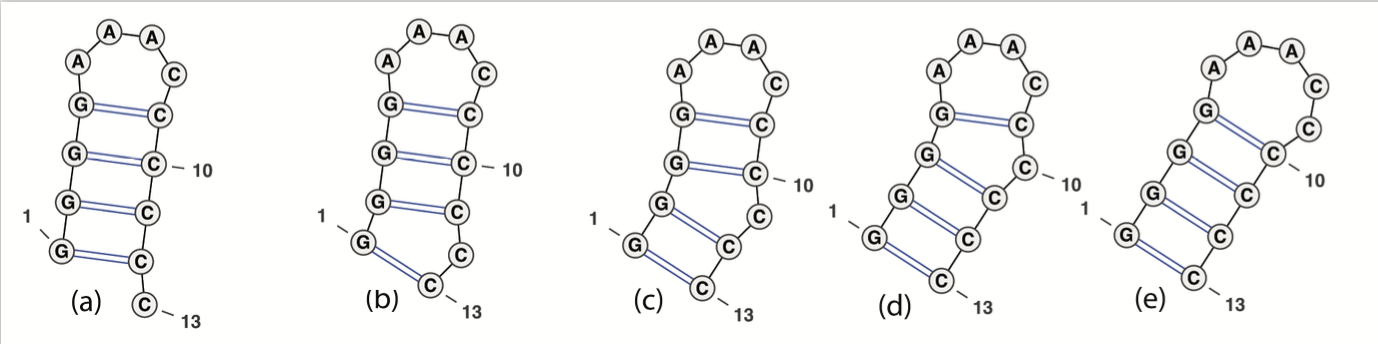}
\caption{Defect diffusion \cite{defectDiffusion}, where a bulge migrates
stepwise to become absorbed in an hairpin loop. The move from structure
(a) to structure (b) is possible by the shift $(1,12) \rightarrow (1,13)$,
the move from (b) to (c) by shift $(2,11) \rightarrow (2,12)$, etc.
Our algorithm properly accounts for such moves with respect to energy
models A,B,C.  Image taken from \cite{Clote.po15}.
}
\label{fig:defectDiffusion}
\end{figure}

\begin{figure}
\centering
\includegraphics[width=0.8\textwidth]{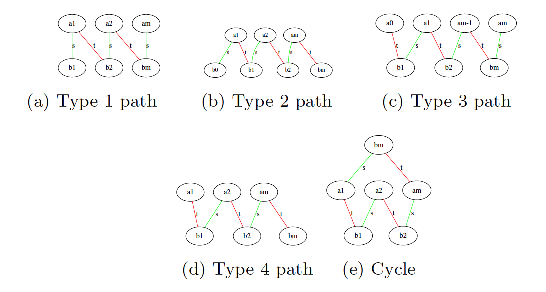}
\caption{All possible maximal length  red-green paths and cycles.
Each equivalence class $X$, as defined in Definition~\ref{def:equivRelation},
can be depicted as a maximal length path or cycle, consisting of those 
positions $x \in [1,n]$ that are connected by alternating base pairs drawn 
from secondary structures $s$ (green) and $t$ (red).  Each path or cycle $X$ 
is depicted in a fashion that the leftmost [resp. rightmost position] 
satisfies the following properties, where $End(s,X)$ [resp. $End(t,X)$]
denotes the set of elements of $X$ that are untouched by $t$ [resp. $s$]: 
for paths of type 1, $|X|$ is even, 
$|End(s,X)]=2$, $b_1=\min(End(s,X))$, $t[a_m]=0$;
for paths of type 2, $|X|$ is odd, 
$|End(s,X)]=1=|End(t,X)|$, $b_0=\min(End(s,X))$, $s[b_m]=0$;
for paths of type 3, $|X|$ is odd, 
$|End(s,X)]=1=|End(t,X)|$, $a_0=\min(End(t,X))$, $t[a_m]=0$;
for paths of type 4, $|X|$ is even, 
$|End(t,X)]=2$, $a_1=\min(End(t,X))$, $s[b_m]=0$;
for paths of type 5, or cycles, $|X|$ must be even, and
$a_1 = \min(X)$.
Note that the appearance of positions in left-to-right order does {\em not}
necessarily respect integer ordering, so the leftmost 
position is not necessarily the minimum $\min(X)$, nor is the rightmost 
position necessarily the maximum $\max(X)$.
}
\label{fig:redBluePathsCycles}
\end{figure}

\begin{figure}
\centering
\includegraphics[width=0.8\textwidth]{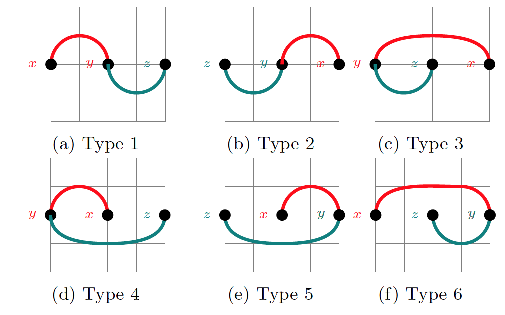}
\caption{All six possible shift moves, in which a base pairs of $s$ (teal) 
that touches a base pairs of $t$ (red) is shifted, thus reducing the base
pair distance $d_{\mbox{\small BP}}(s,t)$ by $2$. 
Each such shift move can uniquely be designated
by the triple $(x,y,z)$, where $y$ is the {\em pivot position} 
(common position to
a base pair in both $s$ and $t$), $x$ is the remaining position in the base pair
in $t$, and $z$ is the remaining position in the base pair in $s$.
}
\label{fig:sixShiftMovesRedGreen}
\end{figure}

\begin{figure}
\centering
\includegraphics[width=0.8\textwidth]{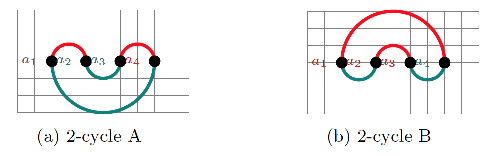}
\caption{Two special types of {\em closed 2-cycles}.
(a) RNA conflict digraph $G=(V,E)$ for secondary structures $t$ and $s$,
where $a_1<a_2<a_3<a_4$ and
$t = \{ (a_1,a_2), (a_3,a_4) \}$, and
$s = \{ (a_1,a_4), (a_2,a_3) \}$.
Nodes of $V = \{ v_1,v_2,v_3,v_4 \}$ are the following:
$v_1 = (a_1,a_2,a_3)$ of type 1, 
$v_2 = (a_3,a_4,a_1)$ of type 5, 
$v_3 = (a_2,a_1,a_4)$ of type 4, and 
$v_4 = (a_4,a_3,a_2)$ of type 2. 
(b) RNA conflict digraph $G=(V,E)$ for secondary structures $t$ and $s$,
where $a_1<a_2<a_3<a_4$ and
$t = \{ (a_1,a_4), (a_2,a_3) \}$ and
$s = \{ (a_1,a_2), (a_3,a_4) \}$.
Nodes of $V = \{ v_1,v_2,v_3,v_4 \}$ are the following:
$v_1 = (a_1,a_4,a_3)$  of type 6,
$v_2 = (a_4,a_1,a_2)$  of type 3,
$v_3 = (a_2,a_3,a_4)$  of type 1,
$v_4 = (a_3,a_2,a_1)$  of type 2.
Since the overlap between any two distinct vertices in (a) and (b) is 2, 
there are \underline{no} edges in $E$ for the conflict digraphs of (a) and (b).
An optimal trajectory from $s$ to $t$ is constructed by removing a base
pair from $s$, performing a shift, and adding the remaining base pair from 
$t$. In each case there are 2 choices for the base pair to remove and
two choices for the shift, so 4 optimal trajectories for each of (a) and
(b). 
}
\label{fig:specialClosed2cycle}
\end{figure}

\begin{figure}
\centering
\includegraphics[width=0.8\textwidth]{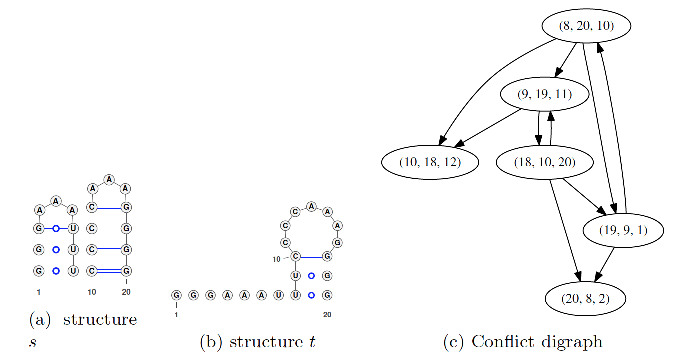}
\caption{Conflict digraph for a 20 nt toy example with sequence
GGGAAAUUUC CCCAAAGGGG, with initial structure $s$ whose free energy
is +0.70 kcal/mol, and target structure $t$ whose free energy is
+ 3.30 kcal/mol. The conflict digraph contains 3 simple cycles:
a first cycle
$\{(8, 20, 10), (9, 19, 11), (18, 10, 20), (19, 9, 1)\}$ of size 4,
a second cycle $\{[(8, 20, 10), (19, 9, 1)\}$ of size 2, and a third
cycle $\{(18, 10, 20), (9, 19, 11)\}$ of size 2.
}
\label{fig:cycleExample}
\end{figure}

\begin{figure}
\centering
\includegraphics[width=0.9\textwidth]{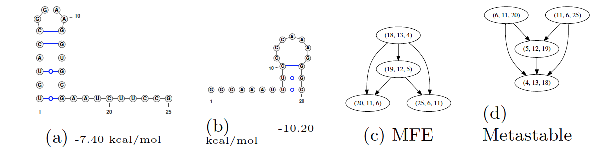}
\caption{Conflict digraphs for the 25 nt bistable switch with sequence
UGUACCGGAA GGUGCGAAUC UUCCG taken from 
Figure 1(b).1 of \cite{Hobartner.jmb03}, in which the authors performed
structural probing by comparative imino proton NMR spectroscopy.
(a) Minimum free energy (MFE) structure having -10.20 kcal/mol.
(b) Alternate metastable structure having next lowest free free energy of
-7.40 kcal/mol. Two lower energy structures exist, having -9.00 kcal/mol
resp. -7.60 kcal/mol; however, each is a minor variant of the MFE structure.
(c) RNA conflict digraph $G=(V,E)$, having directed edges
$(x,y,z) \rightarrow (u,v,w)$ if the (unordered) base pair
$\{ y,z\} \in s$ touches or crosses the (unordered) base 
pair $\{ u,v\} \in t$.  Here, $s$ is
in the metastable structure shown in (b) having -7.40 kcal/mol, while
$t$ is the MFE structure shown in (a) having -10.20 kcal/mol. The
conflict digraph represents a necessary order of application of shift
moves, in order to avoid the creation of base triples or pseudoknots.
Note that the digraph $G$ is
acyclic, but the IP solver must nevertheless be invoked with constraint
($\ddag$) that precludes both vertices $(20,11,6)$ and $(25,6,11)$ from
belonging to the solution $\Vbar$.
(d) RNA conflict digraph $G'=(V',E')$, having similar definition in which
roles of $s$ and $t$ are reversed -- i.e. $MS_2$ folding pathways from 
the MFE structure to the (higher energy) metastable structure.
}
\label{fig:conflictDigraphMicuraSwitch}
\end{figure}

\begin{figure}
\centering
\includegraphics[width=0.8\textwidth]{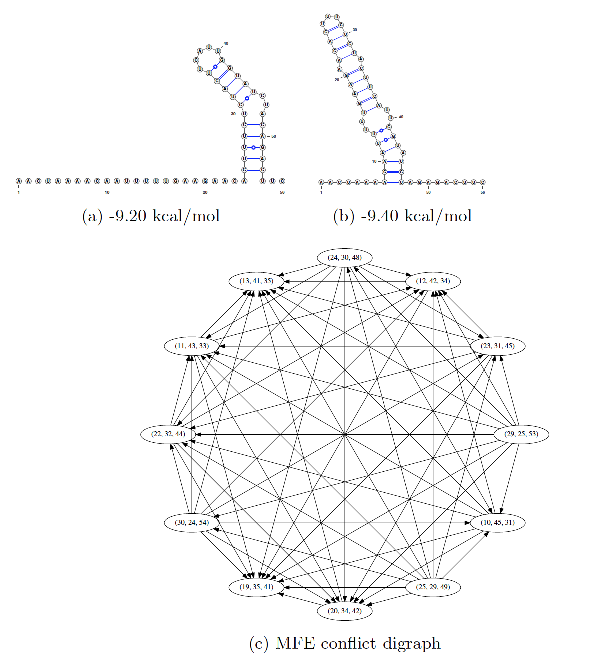}
\caption{Conflict digraph for the 56 nt spliced 
leader RNA from {\em L. collosoma}
with sequence AACUAAAACA AUUUUUGAAG AACAGUUUCU GUACUUCAUU GGUAUGUAGA
GACUUC. (a) Metastable structure having free energy of -9.20 kcal/mol. (b) 
Minimum free energy (MFE) structure having free energy of -9.40 kcal/mol.
(c) The RNA conflict digraph for refolding from metastable $s$ 
to MFE $t$ contains 12 vertices, 61 directed edges and no directed cycles.
Minimum free energy (MFE) and metastable structures in (a) and (b) computed
by Vienna RNA Package \cite{Lorenz.amb11}. Secondary structure images
generated using VARNA \cite{Darty.b09}.
}
\label{fig:lecuyer}
\end{figure}

\begin{figure}
\centering
\includegraphics[width=0.8\textwidth]{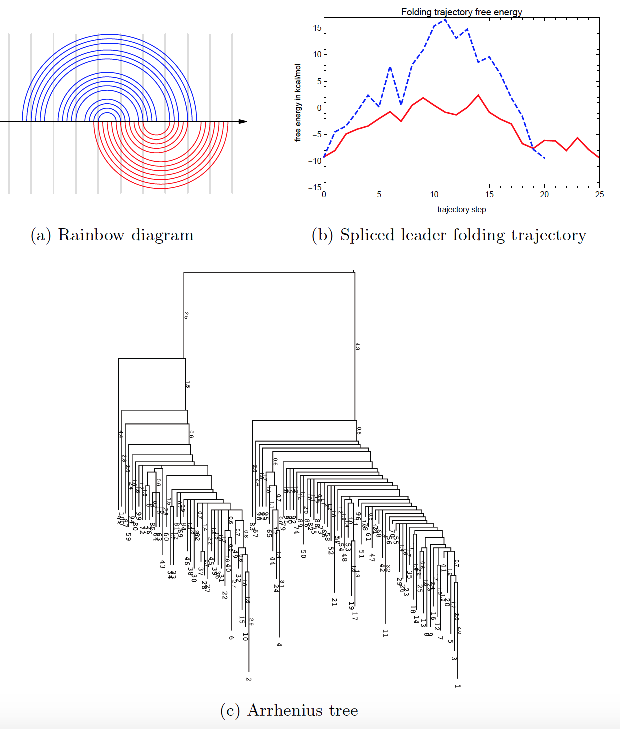}
\caption{(a) Rainbow diagram for spliced 
leader RNA from {\em L. collosoma}, in which the base pairs for the
initial structure $s$ (Figure~\ref{fig:lecuyer}a)
are shown below the line in red, while those for the
target structure $t$ (Figure~\ref{fig:lecuyer}b)
are shown above the line in blue.
(b) Free energies of structures in the shortest $MS_2$ folding trajectory 
for spliced leader are shown by the dotted blue line, while those for the
energy-optimal $MS_2$ folding trajectory are shown in the solid red line.
Algorithm~\ref{algo:MS2path} was used to compute the shortest $MS_2$
trajectory, while the programs {\tt RNAsubopt} 
\cite{wuchtyFontanaHofackerSchuster} and {\tt barriers} \cite{flammHofacker}
were used to compute the energy-optimal folding trajectory.
}
\label{fig:LcollosomaBis}
\end{figure}

\begin{figure}
\centering
\includegraphics[width=0.8\textwidth]{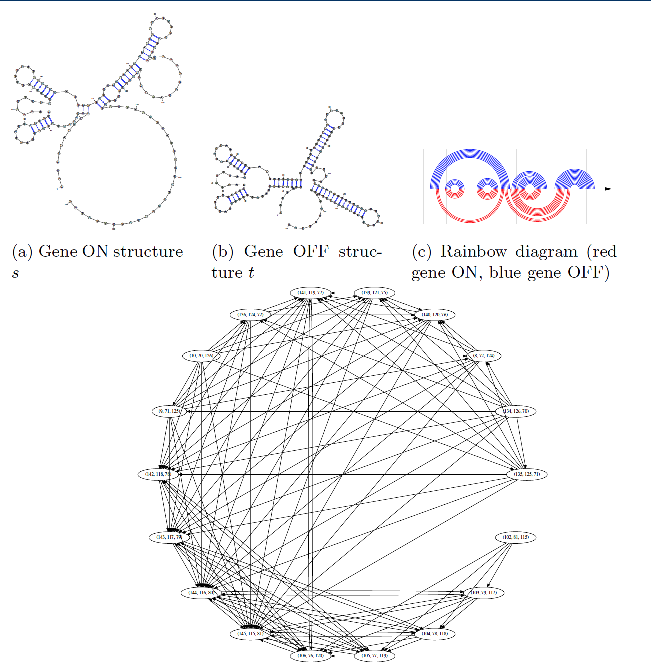}
\caption{Gene ON and gene OFF structures and the RNA conflict digraph for
the 156 nt xanthine phosphoribosyltransferase (xpt) riboswitch from 
{\em B. subtilis} -- 
structures consistent with in-line probing data taken from
Figure 1A of \cite{breaker:Riboswitch2}.
(a) Gene ON structure (default) in absence of free guanine, having
(computed) free energy of -33.11 kcal/mol.
(b) Gene OFF structure when guanine binds cytidine in position 66, having
(computed) free energy of -56.20 kcal/mol (guanine not shown). 
(c) Rainbow diagram with red gene ON structure below line and blue gene OFF
structure above line. Rainbow diagrams allow one to determine by visual
inspection when base pairs touch or cross.
(d) Conflict digraph $G=(V,E)$, containing 18 vertices, 113 directed
edges, and 1806 directed cycles.
Minimum free energy (MFE) and metastable structures in (a) and (b) computed
by Vienna RNA Package \cite{Lorenz.amb11}. Free energy computations using
Turner energy model \cite{Turner.nar10} computed with Vienna RNA Package
\cite{Lorenz.amb11}. Secondary structure images
generated using VARNA \cite{Darty.b09}.
}
\label{fig:conflictDigraphXPT}
\end{figure}

\begin{figure}
\centering
\includegraphics[width=0.8\textwidth]{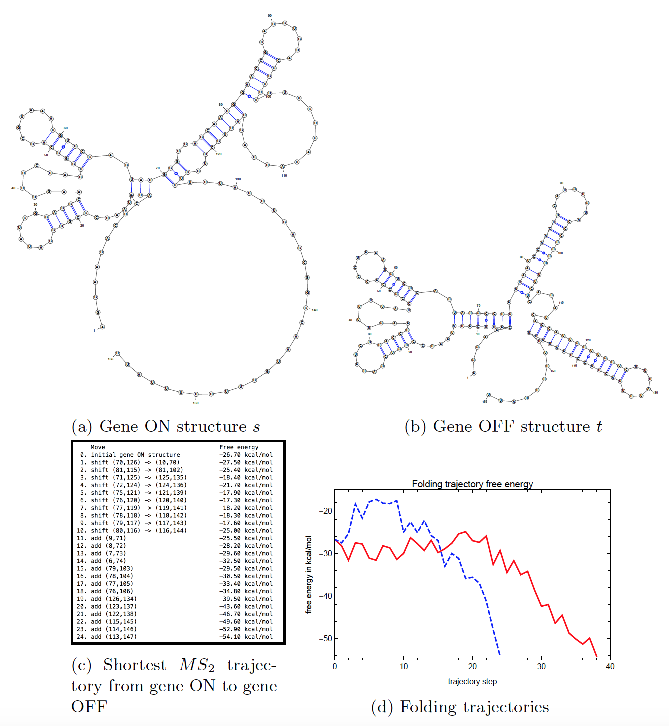}
\caption{(a) Initial gene ON structure (default) in absence of free guanine 
(enlarged image).
(b) Target gene OFF structure when guanine binds cytidine in position 66 
(enlarged image).
(c) Shortest $MS_2$ folding trajectory 
from the gene ON structure $s$ to the gene OFF structure $t$ for the  156 nt
xpt riboswitch from {\em B. subtilis}, described in the caption to
Figure~\ref{fig:conflictDigraphXPT}. 
Note the initial elongation of the P1 helix by the first shift, followed by 
the stepwise removal of the anti-terminator and construction of the terminator
loops by shift 2-10, followed by base pair additions to lengthen the
terminator loop.
Although (a,b) are duplicated from the previous figure,
this is necessary to follow the sequence of moves in the $MS_2$ trajectory.
Secondary structure images generated using VARNA \cite{Darty.b09}.
(d) Free energies of structures in the shortest $MS_2$ folding trajectory 
for xpt are shown by the dotted blue line, while those for the an
energy near-optimal $MS_1$ folding trajectory are shown in the solid red line.
Algorithm~\ref{algo:MS2path} was used to compute the shortest $MS_2$
trajectory, while the program {\tt RNAtabuPath} 
\cite{Dotu.nar10} was
used to compute the energy near-optimal $MS_1$ folding trajectory. 
The size of the 156 nt xpt riboswitch and the fact that the program
{\tt RNAsubopt} would need to generate all secndary structures within
30 kcal/mol of the minimum free energy -54.1 kcal/mol preclude any possibility
that the optimal $MS_2$ trajectory can be computed by application of the
program {\tt barriers} \cite{flammHofacker}. 
}
\label{fig:xptStructuresAndTrajectory}
\end{figure}

\begin{figure}
\centering
\includegraphics[width=0.7\textwidth]{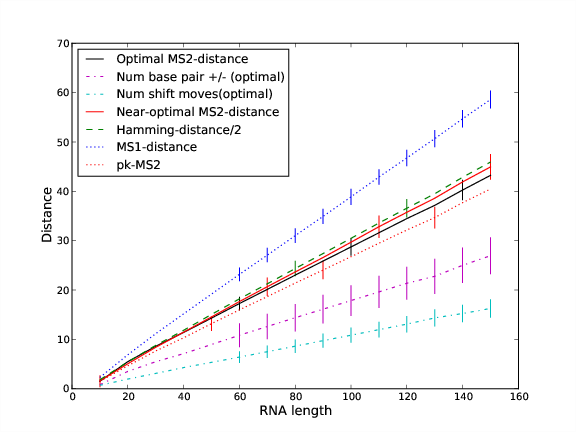}
\caption{Benchmarking statistics for optimal and near-optimal algorithm
to compute minimum length $MS_2$ folding trajectories between random
secondary structures $s,t$ of random RNA sequences of variable lengths.
For each sequence length $n=10,15,20,\cdots,150$ nt, twenty-five random RNA
sequences were generated of length $n$, with probability of $1/4$ for each
nucleotide. For each RNA sequence, twenty secondary structures $s,t$ were 
uniformly randomly generated so that 40\% of the nucleotides are base-paired.
Thus the number of computations per sequence length is thus 
$25 \cdot 190 = 4750$,
so the size of the benchmarking set is $15 \cdot 4750 = 71,250$. 
Using this dataset, the average $MS_2$ distance was computed 
for both the exact IP Algorithm~\ref{algo:MS2path} and the near-optimal
Algorithm~\ref{algo:nearOptimalMS2path}.
In addition to $MS_2$ distance computed by the exact IP and the near-optimal
algorithm, the figure displays $pk-MS_2$ distance (allowing pseudoknots in
intermediate structures) as computed by
Algorithm~\ref{algo:ms2pathLengthArbitraryStr}, 
the $MS_1$ distance (also known as base
pair distance), Hamming distance over 2, and provides a breakdown of the
$MS_1$ distance in terms of the number of base pair addition/removal moves
``num base pair +/- (optimal)'' and the shift moves
``num shift moves (optimal)''.
}
\label{fig:benchmarkingAccuracy}
\end{figure}

\begin{figure}
\centering
\includegraphics[width=0.9\textwidth]{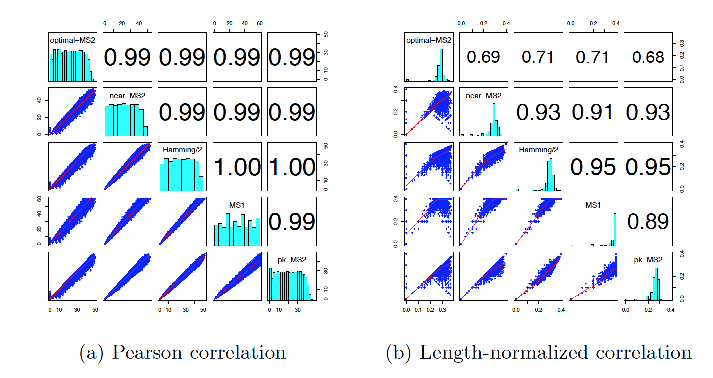}
\caption{Pairwise correlations for optimal $MS_2$ distance,
$pk-MS_2$ distance, near-optimal $MS_2$ distance, Hamming distance divided
by $2$,  and $MS_1$ distance (also called base pair distance).
For each two measures, scatter plots were created for the
$71,250$ many data points from the benchmarking set described in 
Figure~\ref{fig:benchmarkingAccuracy}. Pearson correlation and
{\em normalized} Pearson correlation values computed, where by
normalized, we mean that for each of the $71,250$ data points,
we consider the {\em length-normalized} distance (distance divided by
sequence length). These correlations are statistically significant --
each Pearson correlation values has a $p$-value less than $10^{-8}$.
}
\label{fig:benchmarkingCorrelation}
\end{figure}

\begin{figure}
\centering
\includegraphics[width=0.9\textwidth]{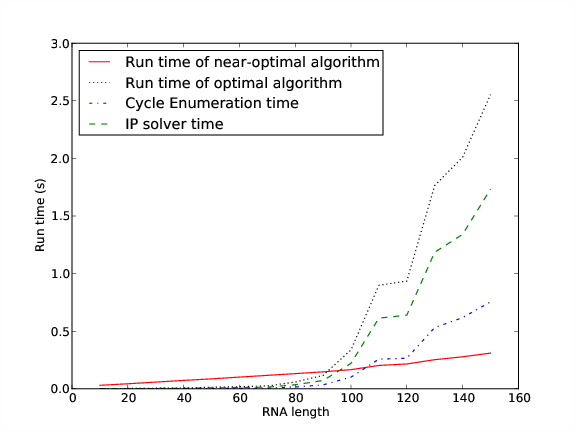}
\caption{Run time for the exact IP (optimal) algorithm~\ref{algo:MS2path}
and the near-optimal algorithm~\ref{algo:nearOptimalMS2path}
to compute minimum length $MS_2$ folding trajectories for the same data
set from previous Figure~\ref{fig:benchmarkingAccuracy}. Each data point
represents the average $\mu \pm \sigma$ where error bars indicate one
standard deviation, taken over  $71,250$ many sequence/structure pairs.
Run time of the optimal algorithm depends on time to perform topological
sort, time to enumerate all directed cycles using our Python implementation
of Johnson's algorithm \cite{Johnson:cycleEnumeration},  and time for 
the Gurobi IP solver (ordered here by increasing time demands).
}
\label{fig:runTimeBenchmark}
\end{figure}

\begin{figure}
\centering
\includegraphics[width=0.9\textwidth]{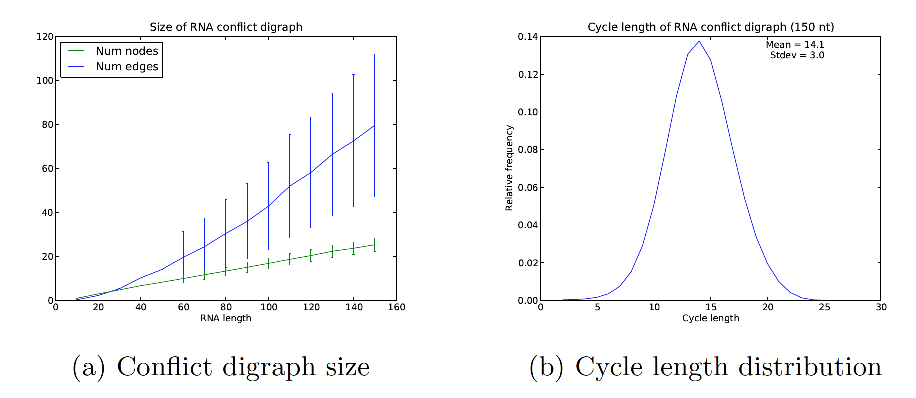}
\caption{{\em (Left)} Average size of vertex sets $V$ and of directed edge sets
$E$ for RNA conflict digraphs $G=(V,E)$ for the data set of described in
Figure~\ref{fig:benchmarkingAccuracy}. Error bars represent $\pm 1$ standard
deviation. Clearly the size of a conflict digraph grows linearly in the length
$n$ of random RNAs ${\bf a}=a_1,\ldots,a_n$,
given random secondary structures $s,t$ having $n/5$ base pairs.
{\em (Right)}
Cycle length distribution for random RNAs ${\bf a}=a_1,\ldots,a_n$ 
of length $n=150$, with randomly chosen secondary structures $s,t$ having
$n/5$ base pairs, using data extracted from the data set described in
Figure~\ref{fig:benchmarkingAccuracy}. For values of $n=50,\ldots,150$,
the cycle length distribution appears approximately normal, although
this is not the case for $n \leq 40$ (data not shown).
}
\label{fig:benchmarkingSizeConflictDigraphAndCycleLengthDistribution}
\end{figure}

\begin{figure}
\centering
\includegraphics[width=0.7\textwidth]{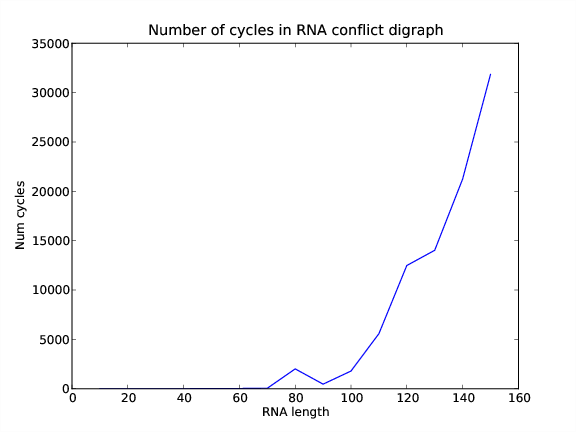}
\caption{Average number of directed cycles as a function of sequence length 
for the data set described in Figure~\ref{fig:benchmarkingAccuracy}. 
For each random RNA sequence ${\bf a}=a_1,\ldots,a_n$ of length $n$, and 
for each pair of random secondary structures $s,t$ of ${\bf a}$ having
$n/5$ base pairs, we computed the total number of directed cycles in the 
conflict digraph $G({\bf a},s,t)$.  The figure suggests that starting at
a threshold sequence length $n$, there is an exponential growth in the number
of directed cycles in the conflict digraph of random sequences of length
$n$. 
}
\label{fig:benchmarkingAvgNumCycles}
\end{figure}

\begin{figure}
\centering
\includegraphics[width=0.9\textwidth]{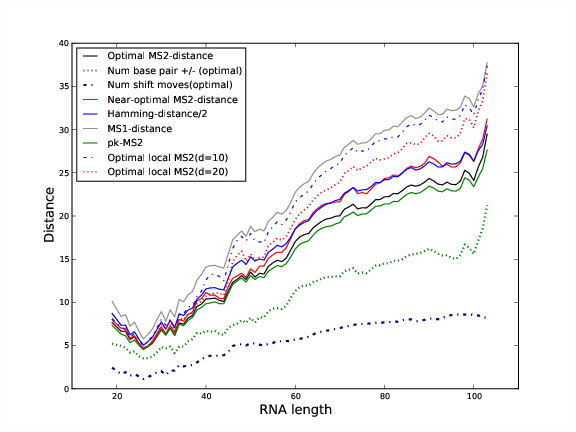}
\caption{Moving averages of distance measures graphed as a function of 
sequence length for the $1311$ many sequences extracted from
Rfam 12.0 (see text). 
Distance measures included
optimal $MS_2$-distance computed
by the exact IP (optimal) algorithm~\ref{algo:MS2path} (where
the number of base pair additions ($+$) or removals ($-$) is indicated,
along with the number of shifts),, 
near-optimal $MS_2$-distance computed
by near-optimal algorithm~\ref{algo:nearOptimalMS2path},
Hamming distance divided by $2$,
$MS_1$ distance aka base pair distance,
pseudoknotted $MS_2$ distance (pk-$MS_2$) computed
from Lemma~\ref{lemma:pkMS2},
optimal local $MS_2$ with parameter $d=10$, and
optimal local $MS_2$ with parameter $d=20$. The latter values were computed
by a variant of the exact IP algorithm~\ref{algo:MS2path} where
shift moves were restricted to be {\em local} with parameter $d$, whereby
base pair shifts of the form
$(x,y) \rightarrow (x,z)$ or $(y,x) \rightarrow (z,x)$ were allowed only
when $|y-z| \leq d$. All moving averages were computed over symmetric
windows of size 9, i.e. $[i-4,i+4]$.
From smallest to largest value, the measures are:
number of shifts in optimal $MS_2$ trajectory $<$
number of base pair additions or deletions ($+/-$) in optimal $MS_2$ 
trajectory $<$
pk-$MS_2$ $<$ $MS_2$ distance $<$ Hamming distance
over $2$ $\approx$ near-optimal $MS_2$ $<$ 
$MS_2$ with locality parameter $d=20$ $<$
$MS_2$ with locality parameter $d=10$ $<$
$MS_1$.
}
\label{fig:distanceMeasureBenchmarkRfam}
\end{figure}

\begin{figure}
\centering
\includegraphics[width=0.9\textwidth]{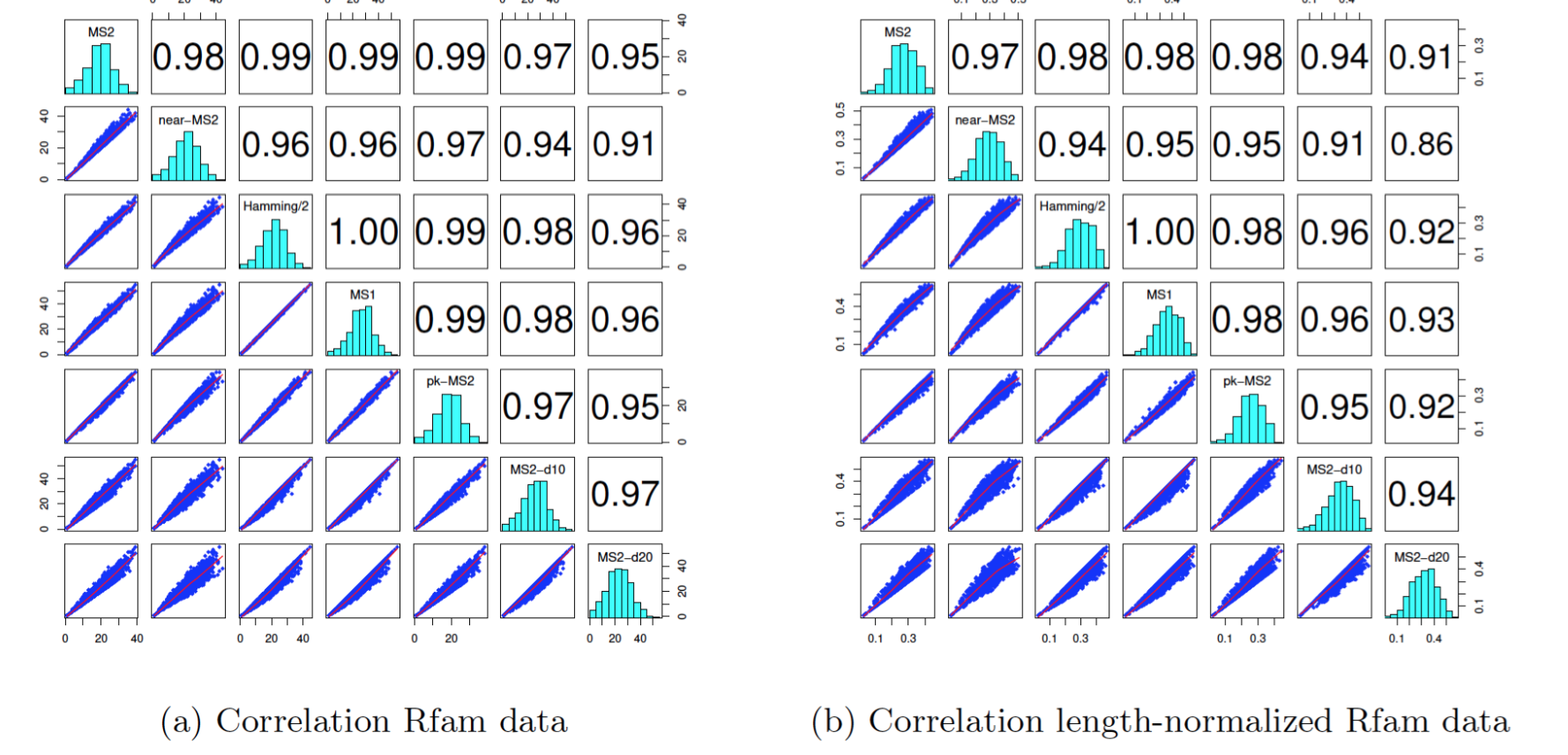}
\caption{Pairwise correlations using Rfam benchmarking data
(see text) for optimal $MS_2$ distance,
$pk-MS_2$ distance, near-optimal $MS_2$ distance, Hamming distance divided
by $2$,  and $MS_1$ distance (also called base pair distance).
For each two measures, scatter plots were created for the
$1311$ many data points.  Pearson correlation and
{\em normalized} Pearson correlation values computed, where by
normalized, we mean that for each of the $1311$ data points,
we consider the {\em length-normalized} distance (distance divided by
sequence length). These correlations are statistically significant --
each Pearson correlation values has a $p$-value less than 
$2.2 \cdot 10^{-16}$.
}
\label{fig:RfamDataCorrelation}
\end{figure}

\begin{figure}
\centering
\includegraphics[width=0.7\textwidth]{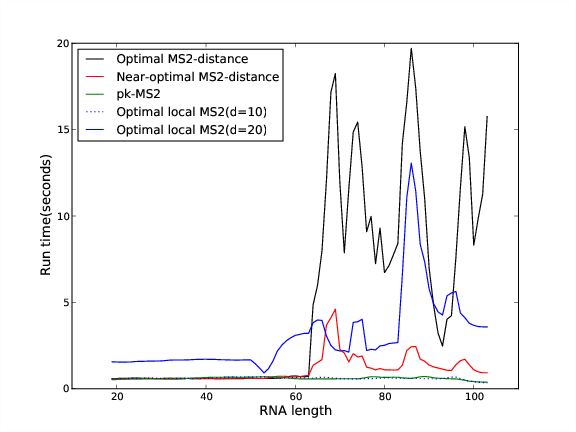}
\caption{Moving averages of run time as a function of 
sequence length for the $1311$ sequences extracted from
Rfam 12.0 (see text). Distance measures considered are
the exact $MS_2$ distance computed by the 
optimal IP Algorithm~\ref{algo:MS2path}, an approximation to the
$MS_2$ distance computed by the near-optimal IP 
Algorithm~\ref{algo:nearOptimalMS2path},
the $pk-MS_2$ distance (allowing pseudoknots in
intermediate structures) as computed by
Algorithm~\ref{algo:ms2pathLengthArbitraryStr},  and two variants
of exact $MS_2$ distance, where shifts are restricted by locality
parameter $d = 10,20$.  These latter values were computed
by the exact IP Algorithm~\ref{algo:MS2path} modified to allow
base pair shifts of the form
$(x,y) \rightarrow (x,z)$ or $(y,x) \rightarrow (z,x)$ only
when $|y-z| \leq d$. 
}
\label{fig:runTimeBenchmarkRfam}
\end{figure}

\begin{figure}
\centering
\includegraphics[width=0.7\textwidth]{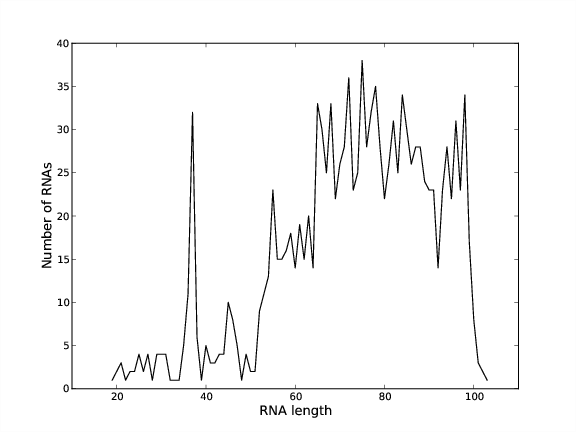}
\caption{Number of sequences as a function of sequence length
for the 1311 sequences extracted from Rfam 12.0 and used in
benchmarking tests (see text).
}
\label{fig:numberRfamSequences}
\end{figure}

\begin{figure}
\centering
\includegraphics[width=0.8\textwidth]{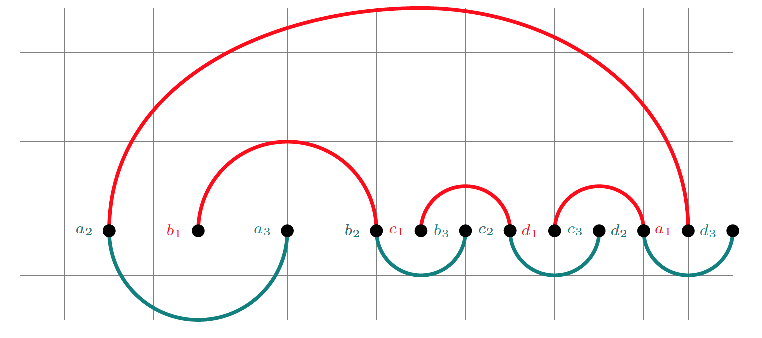}
\caption{RNA conflict digraph $G=(V,E)$ for secondary structures $t$ and $s$,
where
$t = \{ (a_2,a_1), (b_1,b_2),(c_1,c_2), (d_1,d_2) \}$, and
$s = \{ (a_2,a_3), (b_2,b_3),(c_2,c_3), (d_2,d_3) \}$.
The triplet nodes of $V = \{ v_a,v_b, v_c,v_d \}$ are the following:
$v_a = (a_1,a_2,a_3)$ of type 3, shift $v_b = (b_1,b_2,b_3)$ of type 1, shift
$v_c = (c_1,c_2,c_3)$ of type 1, and shift $v_d = (d_1,d_2,d_3)$ of type 1. 
The edges in $E$ are the following:
$v_a \rightarrow v_b$,
$v_b \rightarrow v_c$,
$v_c \rightarrow v_d$,
$v_d \rightarrow v_a$.
The conflict digraph $G=(V,E)$ is 
order-isomorphic to the digraph $G'=(V',E')$, where $V' = \{ 1,2,3,4\}$
and edges are as follows: 
$1 \rightarrow 2 \rightarrow 3 \rightarrow 4 \rightarrow 1$. 
}
\label{fig:4cycle}
\end{figure}

\begin{figure}
\centering
\includegraphics[width=0.8\textwidth]{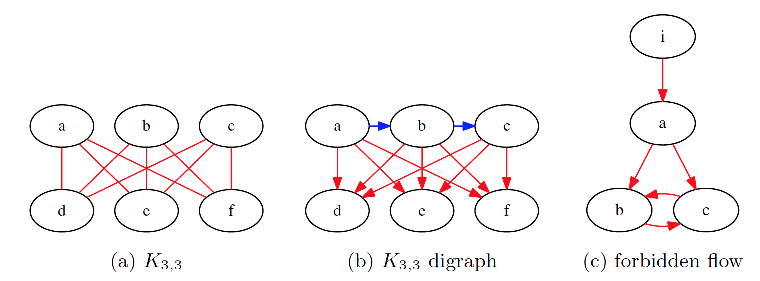}
\caption{(a) Complete bipartite graph $K_{3,3}$. A finite graph 
is planar if and only if it does not contain the forbidden graph $K_{3,3}$
or the complete graph $K_5$ \cite{forbiddengraphK33planar}.
(b) Directed graph realized by the RNA conflict digraph in 
Figure~\ref{fig:RNAconflictGraphForK33}.
It follows that RNA conflict digraphs are not necessarily planar.
(c) Directed graph realized by the RNA conflict digraph in
Figure~\ref{fig:RNAconflictGraphForFlowGraph}. 
A flow graph is {\em reducible} if and only
if it does not contain such a forbidden flow graph, where edges between
nodes may be replaced by arc-disjoint directed paths \cite{flowgraphUllman}. 
It follows that RNA conflict digraphs are not necessarily reducible flow
graphs.
}
\label{fig:forbiddenGraphs}
\end{figure}

\begin{figure}
\centering
\includegraphics[width=0.8\textwidth]{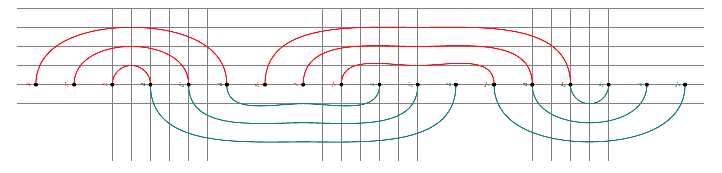}
\caption{RNA conflict digraph that realizes the digraph $K_{3,3}$ depicted in
Figure~\ref{fig:forbiddenGraphs}b, whose undirected red edges represent the
undirected graph $K_{3,3}$ depicted in Figure~\ref{fig:forbiddenGraphs}b.
The nonplanar complete, bipartite digraph $K_{3,3}$, with shift moves
$c,b,a,d,e,f$, all of type 1, in order from left to right -- i.e.
order of positions along the $x$-axis is given by:
$c_1,b_1,a_1,a_2,b_2,c_2$,
$d_1,e_1,f_1,c_3,b_3,a_3$,
$f_2,e_2,d_2,d_3,e_3,f_3$.  Notice that
$a_s$ crosses $b_t$, $c_t$, $d_t$, $e_t$ and $f_t$ so 
$c \leftarrow a$, $b \leftarrow a$ and 
$a \rightarrow d$, $a \rightarrow e$, $a \rightarrow f$;
$b_s$ crosses $c_t$, $d_t$, $e_t$ and $f_t$ so 
$c \leftarrow b$ and
$b \rightarrow d$, $b \rightarrow e$, $b \rightarrow f$;
$c_s$ crosses $d_t$, $e_t$ and $f_t$ so 
$c \rightarrow d$, $c \rightarrow e$, $c \rightarrow f$.
}
\label{fig:RNAconflictGraphForK33}
\end{figure}

\begin{figure}
\centering
\includegraphics[width=0.8\textwidth]{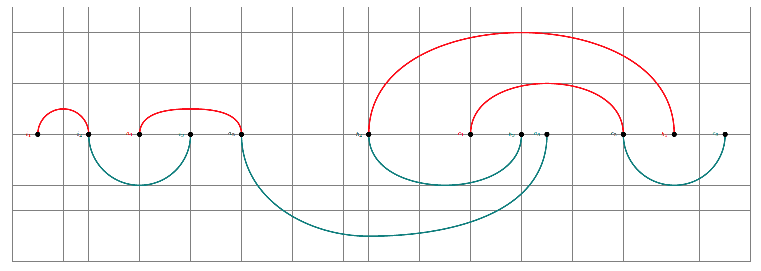}
\caption{Forbidden flow graph with nodes $i=(i_1,i_2,i_3)$,
$a=(a_1,a_2,a_3)$, $b=(b_1,b_2,b_3)$, $c=(c_1,c_2,c_3)$, where
nodes $i,a,c$ are of type 1 and node $b$ is of type 3. Notice that
$i_s$ crosses $a_t$ so $i \rightarrow a$;
$a_s$ crosses $b_t$ so $a \rightarrow b$;
$a_s$ crosses $c_t$ so $a \rightarrow c$; 
$b_s$ crosses $c_t$ so $b \rightarrow c$; 
$c_s$ crosses $b_t$ so $c \leftarrow b$.
}
\label{fig:RNAconflictGraphForFlowGraph}
\end{figure}

\begin{figure}
\centering
\includegraphics[width=0.8\textwidth]{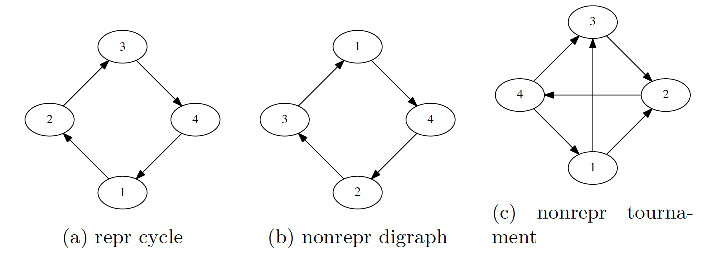}
\caption{(a) Digraph of an ordered 4-cycle, which is representable by
an RNA conflict digraph, as shown in Figure~\ref{fig:4cycle}.
(b) Digraph of an ordered 4-cycle, which is {\em not} representable by
an RNA conflict digraph. Note that digraph (b) is Eulerian, with the 
property that the in-degree of each vertex equals its out-degree.
(c) Digraph of a tournament on 4 vertices, which is
{\em not} representable by an RNA conflict digraph. Digraph (a) is
isomorphic with digraph (b), thus showing that representability is not
preserved under isomorphism. 
Since it is not difficult to show that all 
$2^{ {3 \choose 2} }=8$ tournaments on 3 nodes are representable by
RNA conflict digraphs (data not shown), it follows that digraph (c) is
a minimum sized non-representable tournament, which we 
verified by constraint programming. In general there are 
$2^{ {n \choose 2} }$ many tournaments on $n$.
}
\label{fig:digraphsNotRepresentableAsConflictDigraphs}
\end{figure}

\begin{figure}
\centering
\includegraphics[width=0.8\textwidth]{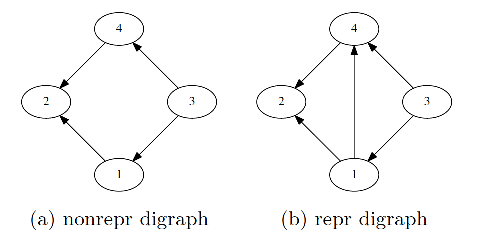}
\caption{Example of a 4-node digraph in (a) that is not representable by an
RNA conflict digraph. However, by adding an edge to digraph (a), we obtain a
representable digraph in (b). Note that digraph (a) is {\em neither} 
order-isomorphic to digraph (b), {\em nor} is there an order-preserving
embedding of digraph (a) into digraph (b). 
}
\label{fig:digraphsNotRepresentableAsConflictDigraphsBis}
\end{figure}

%% file: appendix.tex
\section{Classification of edges in RNA constraint digraphs}
\label{section:edgeClassificationAppendix}

In this section, we describe the collection of all possible directed edges
$v \rightarrow v'$ in which $v.s$ crosses $v'.t$, that can
appear in an RNA conflict digraph, classified as forward, backward or 
2-cycles and according to each type of vertex (see 
Figure~4 for the six types of vertices).
It is straightforward for the reader to imaging additional directed edges
$v \rightarrow v'$ in which $v.s$ touches $v'.t$, so these are not shown.
In addition, we provide the pseudocode for a slow branch-and-bound algorithm
to determine a shortest $MS_2$ trajectory.

Given two secondary structures $s,t$ for the RNA sequence $a_1,\ldots,a_n$,
recall that
notation for a shift move from the (unordered) base pair $\{ x,y \} \in s$
to the (unordered) base pair $\{ y,z \} \in t$ is given by the triple
$(z,y,x)$, where the middle coordinate $y$ is the {\em pivot position},
common to both base pairs $\{x,y\}\in s$ and $\{y,z\} \in t$, while the
first [resp. last] coordinate $z$ [resp. $x$]
is the remaining position from the base pair
$\{y,z\} \in t$ [resp.  $\{x,y\}\in s$]. A directed edge is given 
from shift move $(x,y,z)$ to shift move $(u,v,w)$ if the base pair
$\{ y,z\} \in s$ from the first shift move {\em crosses} with the
base pair $\{u,v\}\in t$ from the second shift move; i.e.
$\min(u,v) < \min(y,z) < \max(u,v) < \max(y,z)$ or
$\min(y,z) < \min(u,v) < \max(y,z) < \max(u,v)$. The reason for the
directed edge is that if the second shift $(u,v,w)$ is applied before
the first shift $(x,y,z)$, then a pseudoknot (crossing) would be created;
it follows that the first shift must be applied before the second shift.

Edges may be forward (left-to-right) or backward (right-to-left), depending
on whether the {\em pivot} position of the first shift is (strictly) less than
or (strictly) greater than the {\em pivot} position of the second shift.
This section does not list similar examples, where the (unordered) base pair
$\{ y,z \} \in s$ from the first shift move {\em touches} the (unordered)
base pair $\{ u,v\}\in t$  from the second shift move, as such examples are
clear from Figure~3 of the main text.

\subsection{Forward Edges}
\label{section:appendixForwardEdges}


\hspace{1cm}

\subsection{Backward Edges}
\label{section:appendixBackwardEdges}

%
\hspace{1cm}

\subsection{2-Cycles}
\label{section:appendix2cycles}

%
\hspace{1cm}

\subsection{Summary tables of shift moves edges}
\label{section:appendixSummaryTables}

Table~\ref{table:bidirectionalArcs} presents a count of all 12 possible
bidirectional edges, while 
Table~\ref{table:forwardArcs} [resp.  Table~\ref{table:backArcs}]
presents a count of all 34 possible forward [resp. back] directed edges.
Here, by {\em bidirectional edge} between nodes $x$ and $y$, we mean 
the existence of directed edges $x \rightarrow y$ and $y \rightarrow x$.
Figures in Sections \ref{section:appendix2cycles},
\ref{section:appendixForwardEdges} and \ref{section:appendixBackwardEdges}
depict all of these these directed edges.

\begin{table}	
\centering
\begin{tabular}{||c|c|c|c|c|c||}
\hline \hline 
1 $\leftrightarrow$ 5 & 2 $\leftrightarrow$ 6 & 3 $\leftrightarrow$ 5 & 4 $\leftrightarrow$ 6 & 3 $\leftrightarrow$ 1 & 4 $\leftrightarrow$ 2 \\
\hline\hline
\end{tabular}
\caption{All 6 possible bidirectional edges, or 2-cycles. Note that
$1 \leftrightarrow 5$ is distinct from
$5 \leftrightarrow 1$, since the pivot point from the left node must be
less than that from the right node in our notation.  
Here, by {\em bidirectional edge} between nodes $x$ and $y$, we mean 
the existence of directed edges $x \rightarrow y$ and $y \rightarrow x$.
}
\label{table:bidirectionalArcs}
\end{table}

\begin{table}	
\centering
\begin{tabular}{||cr|cr|cr|cr|cr|cr||}	
\hline\hline
{\sc edge} & num &
{\sc edge} & num &
{\sc edge} & num &
{\sc edge} & num &
{\sc edge} & num &
{\sc edge} & num \\
\hline
1 $\rightarrow$ 1	&	2		&	2 $\rightarrow$ 1	&	1		&	3 $\rightarrow$ 1	&	1		&	4 $\rightarrow$ 1	&	2		&	5 $\rightarrow$ 1	&	1		&	6 $\rightarrow$ 1	&	0		\\
1 $\rightarrow$ 2	&	1		&	2 $\rightarrow$ 2	&	0		&	3 $\rightarrow$ 2	&	1		&	4 $\rightarrow$ 2	&	1		&	5 $\rightarrow$ 2	&	0		&	6 $\rightarrow$ 2	&	0		\\
1 $\rightarrow$ 3	&	1		&	2 $\rightarrow$ 3	&	0		&	3 $\rightarrow$ 3	&	1		&	4 $\rightarrow$ 3	&	1		&	5 $\rightarrow$ 3	&	0		&	6 $\rightarrow$ 3	&	0		\\
1 $\rightarrow$ 4	&	0		&	2 $\rightarrow$ 4	&	0		&	3 $\rightarrow$ 4	&	0		&	4 $\rightarrow$ 4	&	0		&	5 $\rightarrow$ 4	&	0		&	6 $\rightarrow$ 4	&	0		\\
1 $\rightarrow$ 5	&	1		&	2 $\rightarrow$ 5	&	1		&	3 $\rightarrow$ 5	&	0		&	4 $\rightarrow$ 5	&	1		&	5 $\rightarrow$ 5	&	1		&	6 $\rightarrow$ 5	&	0		\\
1 $\rightarrow$ 6	&	2		&	2 $\rightarrow$ 6	&	1		&	3 $\rightarrow$ 6	&	1		&	4 $\rightarrow$ 6	&	2		&	5 $\rightarrow$ 6	&	1		&	6 $\rightarrow$ 6	&	0		\\
\hline	
{\bf 1 $\Rightarrow \ast$}	&	{\bf 7}	&	{\bf 2 $\Rightarrow \ast$}	&	{\bf 3}	&	{\bf 3 $\Rightarrow \ast$}	&	{\bf 4}	&	{\bf 4 $\Rightarrow \ast$}	&	{\bf 7}	&	{\bf 5 $\Rightarrow \ast$}	&	{\bf 3}	&	{\bf 6 $\Rightarrow \ast$}	&	{\bf 0}	 \\
\hline\hline
\end{tabular}	
\caption{All 24 possible forward edges and their number. Here only shift moves
of the form $(x,y,z) \rightarrow (u,v,w)$ are considered, where the (unordered)
base pair $\{y,z\} \in s$ crosses the (unordered) base pair $\{ u,v\} \in t$,
where $y<v$.
}
\label{table:forwardArcs}
\end{table}	

\begin{table}	
\centering
\begin{tabular}{||cr|cr|cr|cr|cr|cr||}	
\hline\hline
{\sc edge} & num &
{\sc edge} & num &
{\sc edge} & num &
{\sc edge} & num &
{\sc edge} & num &
{\sc edge} & num \\
\hline
1 $\leftarrow$ 1	&	0	&	2 $\leftarrow$ 1	&	1	&	3 $\leftarrow$ 1	&	1	&	4 $\leftarrow$ 1	&	1	&	5 $\leftarrow$ 1	&	0	&	6 $\leftarrow$ 1	&	0	\\
 1 $\leftarrow$ 2	&	1	&	 2 $\leftarrow$ 2	&	2	&	3 $\leftarrow$ 2	&	2	&	4 $\leftarrow$ 2	&	1	&	5 $\leftarrow$ 2	&	0	&	6 $\leftarrow$ 2	&	1	\\
 1 $\leftarrow$ 3	&	0	&	 2 $\leftarrow$ 3	&	0	&	3 $\leftarrow$ 3	&	0	&	4 $\leftarrow$ 3	&	0	&	5 $\leftarrow$ 3	&	0	&	6 $\leftarrow$ 3	&	0	\\
1 $\leftarrow$ 4	&	0	&	2 $\leftarrow$ 4	&	1	&	3 $\leftarrow$ 4	&	1	&	4 $\leftarrow$ 4	&	1	&	5 $\leftarrow$ 4	&	0	&	6 $\leftarrow$ 4	&	0	\\
1 $\leftarrow$ 5	&	1	&	2 $\leftarrow$ 5	&	2	&	3 $\leftarrow$ 5	&	2	&	4 $\leftarrow$ 5	&	1	&	5 $\leftarrow$ 5	&	0	&	6 $\leftarrow$ 5	&	1	\\
 1 $\leftarrow$ 6	&	1	&	2 $\leftarrow$ 6	&	1	&	3 $\leftarrow$ 6	&	1	&	4 $\leftarrow$ 6	&	0	&	5 $\leftarrow$ 6	&	0	&	6 $\leftarrow$ 6	&	1	\\
\hline	
{\bf 1 $\Leftarrow \ast$}	&	{\bf 3}	&	{\bf 2 $\Leftarrow \ast$}	&	{\bf 7}	&	{\bf 3 $\Leftarrow \ast$}	&	{\bf 7}	&	{\bf 4 $\Leftarrow \ast$}	&	{\bf 4}	&	{\bf 5 $\Leftarrow \ast$}	&	{\bf 0}	&	{\bf 6 $\Leftarrow \ast$}	&	{\bf 3}	\\
\hline\hline
\end{tabular}
\caption{All 24 possible backward edges and their number. Here only shift moves
of the form $(x,y,z) \leftarrow (u,v,w)$ are considered, where the (unordered)
base pair $\{v,w\} \in s$ crosses the (unordered) base pair $\{ x,y \} \in t$,
where $y<v$.
}
\label{table:backArcs}
\end{table}

\section{Minimal length pk-$MS_2$ folding pathways}
\label{section:pathsOneToFive}

This section provides details on the simple algorithms for pk-$MS_2$ minimum
length folding pathways for each of the five types of paths depicted in
Figure~\ref{fig:redBluePathsCycles}. 
If $s$ and $t$ are (possibly pseudoknotted) structures on $[1,n]$, and 
$X \subseteq [1,n]$ is an equivalence class, then define the 
{\em restriction} of $s$ [resp. $t$] to $X$, denoted by
$s \upharpoonright{X}$ [resp. $t \upharpoonright{X}$], to be the
set of base pairs $(i,j)$ in $s$ [resp. $t$] such that $i,j \in X$.
Each path or cycle in $A \cup B$ can be subdivided into the following
five cases. Each equivalence class can be classified as one of five types
of paths, depicted in Figure~\ref{fig:redBluePathsCycles} described below.
For this classification, we need to define
$End(s,X) = \{ x\in X: t[x]=0\}$  and
$End(t,X) = \{ x\in X: s[x]=0\}$ -- i.e. $End(s,X)$ [resp. $End(t,X)$]
is the set of elements $x$ of $X$ that belong to a base pair in $s$ [resp. $t$],
but the path cannot be extended because $x$ is not touched by a base pair
from $t$ [resp. $s$].
For each type of path $X$, we present a (trivial) algorithm that returns the
shortest $MS_2$ folding trajectory from
$s \upharpoonright{X}$ to $t \upharpoonright{X}$. Additionally, we determine
the relation between the pseudoknotted $MS_2$ distance between
$s \upharpoonright{X}$ and $t \upharpoonright{X}$, denoted
$d^X_{pk-MS_2}(s,t)$, as well as the Hamming distance, denoted $d_H^X(s,t)$.

An equivalence class 
$X$ of size $m$ is defined to be a {\em path of type 1}, if
$m$ is even, so path length is odd, and $|End(s,X)|=2$. 
Let $b_0 = \min(End(s,X))$ and for $1 \leq i < m/2$, define
$a_{i+1}= s[b_{i}]$ and $b_{i}=t[a_{i}]$, as shown in 
Figure~\ref{fig:redBluePathsCycles}a. A minimum length sequence of 
$MS_2$ moves to transform $s\upharpoonright X$ into $t\upharpoonright X$ 
is given by the following:
\medskip

\noindent
{\bf Path 1 subroutine}
\hfill\break
1. remove $\{ b_{m/2},a_{m/2}\}$ from $s$ \\
2. for $(m/2)-1$ down to $1$\\
3. {\hskip 0.5cm} shift base pair $(b_i,a_{i})$ to $(a_i,b_{i+1})$
\medskip

\noindent
An alternate procedure would be to remove the first base pair  $\{b_1,a_1\}$
and perform shifts from left to right.
Notice that if $m=|X|=2$, then a path of type 1 is simply a base pair
with the property that neither $i$ nor $j$ is touched by $t$.  For arbitrary
$m$, $d^X_{pk-MS_2}(s,t) = \frac{d^X_H(s,t)}{2}$.
The Hamming distance $d^X_{H}(s,t)=m$, and
$d^X_{pk-MS_2}(s,t) = m/2$, so $d^X_{pk-MS_2}(s,t) = 
\lfloor \frac{d^X_H(s,t)}{2} \rfloor$. 
Moreover,
$d^X_{pk-MS_2}(s,t) = \max( |s\upharpoonright X|, |t\upharpoonright X|)$.

An equivalence class 
$X$ of size $m$ is defined to be a {\em path of type 2}, if
$m$ is odd, so path length is even, and $|End(s,X)|=1=|End(t,X)|$, 
and $\min(End(s,X))<\min(End(t,X))$.
Let $b_0 = \min(End(s,X))$ and
for $1 \leq i \leq \lfloor m/2 \rfloor$, define
$a_{i+1}= s[b_{i}]$ and $b_{i}=t[a_{i}]$, as shown in 
Figure~\ref{fig:redBluePathsCycles}b. A minimum length sequence of 
$MS_2$ moves to transform $s\upharpoonright X$ into $t\upharpoonright X$ 
is given by the following:
\medskip

\noindent
{\bf Path 2 subroutine}
\hfill\break
1. for $i=\lfloor m/2 \rfloor$ down to $1$ \\
2. {\hskip 0.5cm} shift base pair $\{b_{i-1},a_{i}\}$ to $\{ a_i,b_i \}$
\medskip

\noindent
The Hamming distance $d^X_{H}(s,t)=m$, and
$d^X_{pk-MS_2}(s,t) = \lfloor m/2 \rfloor$, so $d^X_{pk-MS_2}(s,t) = 
\lfloor \frac{d^X_H(s,t)}{2} \rfloor$.
Moreover,
$d^X_{pk-MS_2}(s,t) = \max( |s\upharpoonright X|, |t\upharpoonright X|)$.

An equivalence class 
$X$ of size $m$ is defined to be a {\em path of type 3}, if
$m$ is odd, so path length is even, and $|End(s,X)|=1=|End(t,X)|$, 
and $\min(End(t,X))<\min(End(s,X))$.  Let $a_0 = \min(End(t,X))$ and
for $1 \leq i \leq \lfloor m/2 \rfloor$, define
$b_{i}= t[a_{i-1}]$ and $a_{i}=s[b_{i}]$, as shown in 
Figure~\ref{fig:redBluePathsCycles}c. A minimum length sequence of 
$MS_2$ moves to transform $s\upharpoonright X$ into $t\upharpoonright X$ 
is given by the following:
\medskip

\noindent
{\bf Path 3 subroutine}
\hfill\break
1. for $i=1$ to $\lfloor m/2 \rfloor$\\
2. {\hskip 0.5cm} shift base pair $\{b_{i},a_{i}\}$ to $\{ a_{i-1},b_i \}$
\medskip

\noindent
The Hamming distance $d^X_{H}(s,t)=m$, and
pk-$MS_2$ distance $d^X_{pk-MS_2}(s,t) = \lfloor m/2 \rfloor$, so 
$d^X_{pk-MS_2}(s,t) = \lfloor \frac{d^X_H(s,t)}{2} \rfloor$.
Moreover,
$d^X_{pk-MS_2}(s,t) = \max( |s\upharpoonright X|, |t\upharpoonright X|)$.

An equivalence class 
$X$ of size $m$ is defined to be a {\em path of type 4}, if
$m$ is even, so path length is odd, and $|End(t,X)|=2$.
Let $a_1 = \min(End(t,X))$ and
for $2 \leq i < m/2$, define
$a_{i+1}= s[b_{i}]$ and for
$1 \leq i \leq m/2$, define
$b_{i}=t[a_{i}]$, as shown in 
Figure~\ref{fig:redBluePathsCycles}d. A minimum length sequence of 
$MS_2$ moves to transform $s\upharpoonright X$ into $t\upharpoonright X$ 
is given by the following:
\medskip

\noindent
{\bf Path 4 subroutine}
\hfill\break
1. for $i=1$ to $m/2-1$\\
2. {\hskip 0.5cm} shift base pair $\{b_{i},a_{i+1}\}$ to $\{ a_i,b_i \}$\\
3. add base pair $\{ a_{m/2},b_{m/2} \}$ 
\medskip

\noindent
Notice that if $m=2$, then a path of type 4 is simply a base pair
$(i,j) \in t$, with the property that neither $i$ nor $j$ is touched by $s$.
The Hamming distance $d^X_{H}(s,t)=m$, and
$d^X_{pk-MS_2}(s,t) = m/2$, so $d^X_{pk-MS_2}(s,t) = \frac{d^X_H(s,t)}{2}$.
Moreover,
$d^X_{pk-MS_2}(s,t) = \max( |s\upharpoonright X|, |t\upharpoonright X|)$.

An equivalence class 
$X$ of size $m$ is defined to be a {\em path of type 5}, if it is a cycle,
i.e. each element $x \in X$ is touched by both $s$ and $t$.
Since base triples are not allowed due to 
condition 2 of Definition~\ref{def:secStr}, cycles have only even length,
and so $|X|$ is also even.  Let $a_1 = \min(X)$, and for
$1 \leq i \leq m/2$, define $b_{i}=t[a_{i}]$, and for
$2 \leq i \leq m/2$, define $a_{i}=s[b_{i-1}]$, as shown in 
Figure~\ref{fig:redBluePathsCycles}e. A minimum length sequence of 
$MS_2$ moves to transform $s\upharpoonright X$ into $t\upharpoonright X$ 
is given by the following:
\medskip

\noindent
{\bf Path 5 subroutine}
\hfill\break
1. remove base pair $\{ b_{m/2},a_1 \}$\\
2. for $i=1$ to $m/2-1$\\
3. {\hskip 0.5cm} shift base pair $\{b_{i},a_{i+1}\}$ to $\{ a_i,b_i \}$\\
4. add base pair $\{ a_{m/2},b_{m/2} \}$ 
\medskip

\noindent
The Hamming distance $d^X_{H}(s,t)=m$, and
$d^X_{pk-MS_2}(s,t) = m/2+1$, so $d^X_{pk-MS_2}(s,t) = 
\lfloor \frac{d^X_H(s,t)}{2} \rfloor +1$.  Moreover,
$d^X_{pk-MS_2}(s,t) = \max( |s\upharpoonright X|, |t\upharpoonright X|)$.
Note that any base pair could
have initially been removed from $s$, and by relabeling the 
remaining positions, the same algorithm would apply.

In summary, pk-$MS_2$ distance between 
$s \upharpoonright{X}$ and $t \upharpoonright{X}$
for any maximal path (equivalence class) $X$
is equal to Hamming distance
$\lfloor \frac{d_H(s \upharpoonright{X},t \upharpoonright{X})}{2}
\rfloor$; in contrast,
pk-$MS_2$ distance between $s \upharpoonright{X}$ and $t \upharpoonright{X}$
for any cycle $X$ is equal to 
$\lfloor \frac{d_H(s \upharpoonright{X},t \upharpoonright{X})}{2} \rfloor +1$.  
It follows that $d_{pk-MS_2}(s,t)=\lfloor \frac{d_H(s,t)}{2} \rfloor$ 
if and only if there are no type 5 paths, thus establishing 
equation~(\ref{eqn:relationBetweenHammingAndPKMS2distance}.

Now let $B_1$ [resp. $B_2$] denote the set of positions of all
type 1 paths [resp. type 4 paths] of length 1 -- i.e. positions
incident to isolated green [resp. red] edges that correspond to
base pairs $(i,j) \in s$ where $i,j$ are not touched by $t$
[resp. $(i,j) \in t$ where $i,j$ are not touched by $s$]. As well,
let $B_0$ designate the set of positions in $B$ not in either $B_1$ or 
$B_2$. Note that $B_1 \subseteq B$ and $B_2 \subseteq B$, and that formally
\begin{align}
\label{eqn:defB0bis}
B_0 &= B - (B_1 \cup B_2)  \\
\label{eqn:defB1bis}
B_1 &= \{ i \in [1,n]: \exists j \left[ \{ i,j \} \in s,
t(i)=0=t(j) \right] \\
\label{eqn:defB2bis}
B_2 &= \{ i \in [1,n]: \exists j \left[ \{ i,j \} \in t,
s(i)=0=s(j) \right] 
\end{align}
Note that $B_1$ and $B_2$ have an even number of elements, and that
all elements of $B-B_1-B_2$ are incident to a terminal edge of
a path of length 2 or more. Correspondingly, define $BP_1$ and $BP_2$
as follows:
\begin{align}
\label{eqn:defBP1bis}
BP_1 &= \{ (i,j) \in s: t[i]=0=t[j] \}\\
\label{eqn:defBP2bis}
BP_2 &= \{ (i,j) \in t: s[i]=0=s[j] \}
\end{align}
Note that $|BP_1| = |B_1|/2$ and $|BP_2| = |B_2|/2$. 
The following is a restatement of Lemma~\ref{lemma:pkMS2}.

\begin{lemma}
Let $s,t$ be two arbitrary pseudoknotted structures for the RNA sequence
$a_1,\ldots,a_n$, and let $X_1,\ldots,X_m$ be the
equivalence classes with respect to equivalence relation $\equiv$ on
$A \cup B_0 = [1,n]-B_1-B_2-C-D$. 
Then the pk-$MS_2$ distance between $s$ and $t$ is equal to
\begin{align*}
|BP_1| + |BP_2| + \sum_{i=1}^m \max\Big( |s \upharpoonright{X_i}|,
|t \upharpoonright{X_i}| \Big)
\end{align*}
Alternatively, if $X_1,\ldots,X_m$ are the equivalence classes on
$A \cup B = [1,n]-C-D$, then
\begin{align*}
d_{pk-MS_2}(s,t) &=\sum_{i=1}^m \max\Big( |s \upharpoonright{X_i}|,
|t \upharpoonright{X_i}| \Big)
\end{align*}
\end{lemma}

\section{Branch-and-bound algorithm}
\label{section:branchAndBoundAlgorithm}

In this section, we present pseudocode for an exact and exhaustive 
branch-and-bound search strategy \cite{cormen} to
determine a shortest $MS_2$ folding trajectory between two
given secondary structures of a given RNA sequence.

First base pairs in $BP_1$ and $BP_3$ are removed from $s$ to obtain the root structure.  Starting from the root we perform removal and shift of base pairs until the target or an empty structure is achieved. For each state $S_i$, we have the distance $dist$ from the root to $S_i$. Furthermore, we compute a lower bound $lb$ for the $MS_2$ distance from $S_i$ to the target structure by calculating the $MS_2$ distance allowing pseudoknots. The optimistic $MS_2$ distance from the source to target is estimated as $dist+lb$. If the optimistic distance is higher than the incumbent $MS_2$ distance computed so far, the search is stopped and a different state is considered. In order to prune the search tree more effectively, base pairs that cause more conflicts are considered first. Finally, if an empty structure is achieved we compute the list of remaining moves for $s$.

\begin{algorithm}
\caption{\small
Branch and bound algorithm for $MS_2$ distance\newline \textbf{Input:} Two RNA secondary structures $s$ and $t$ with length $n$  \newline \textbf{Output:} Minimum number of $MS_2$ moves in the path from $s$ to $t$}\label{bb}

\begin{small}
\begin{algorithmic}[1]
\State $BP_1$ is the set of base pairs $(i, j) \in s$ such that neither $i$ nor $j$ is touched by any base pair of $t$ \newline
$BP_2$ is the set of base pairs $(i, j) \in t$ such that neither $i$ nor $j$ is touched by any base pair of $s$ \newline
$BP_3$ is the set of base pairs $(i, j) \in t,s$ \newline
$V = \{(x,y,z) : 1\leq x,y,z \leq n; t[ b ]=a; s[ b ]=c \}$ 
\State remove base pairs in $BP_1 \cup BP_3$ from $s$
\State remove base pairs in $BP_2 \cup BP_3$ from $t$
\State for $v \in V$ compute $n_v$ the the number of crossings between $v.s$ and all base pairs in $t$ 
\State sort $V$ in decreasing order by $n_v$
\State $best = $ base pair distance between $s$ and $t$ \Comment in worst case only $MS_1$ moves are used
\State $curdist=0$ \Comment distance between $s$ and current state $cs$
\State define a data structure $state = \{s,t,dist,lb,rm,ad,sh\}$ representing a node in the search tree. \newline $s$ and $t$ represent the structures in the current state; $dist$ is the $MS_2$ distance from the root;$lb$ is the lower bound for the $MS_2$ distance from $s$ to $t$ for all paths passing through the current state $cs$; $rm$, $ad$ and $sh$ are respectively the lists of removals, additions and shifts performed from the root to obtain the current state.

\State $root = state(s,s,best,0,[~],[~],[~])$
\State define priority queue $Q$ containing states $S_1,...,S_m$, ordered by $S_i.lb$ 
\State $Q = \{root\}$ 
\While {($Q$ is not empty)}
\State $cs = Q.pop()$ \Comment state with smallest lower bound $cs.lb$ will be poped
\If {($cs.lb < best$)}
\If {($s$ has no base pairs)} \Comment current state is a leaf
\State $R =$ set of base pairs in $t$
\State $cs.ad = cs.ad \cup R$
\State $cs.dist = cs.dist + |R| $
\If {($cs.dist < best$)}
\State $Sol = cs $
\State $best = cs.dist$
\EndIf
\ElsIf {($t$ has no base pairs)} \Comment current state is a leaf
\State $R =$ set of base pairs in $s$ 
\State $cs.rm = cs.rm \cup R$
\State $cs.dist = cs.dist+ |R| $
\If {($cs.dist < best$)}
\State $Sol = cs $
\State $best = cs.dist$
\EndIf
\Else \Comment current state is an internal node
\State let $M$ be the list of possible shift moves and removals that can be applied to current state $cs$
\For{$m \in M$}
\State apply move $m$ to $s$
\State $l =$ optimal number of moves to go from $s$ to $t$ passing through $cs+m$ and allowing pseudoknots
\If {$(l<best)$}
\If {($m$ is a shift)}
\State $nt = t - m.t$ \Comment remove the resolved base pair from $t$
\State $ns = s - m.s$ \Comment remove the resolved base pair from $s$
\State $newState = state(ns,nt,cs.dist + 1, l, cs.rm , cs.ad, cs.sh+ m )$
\ElsIf {($m$ is a removal)}
\State $ns = s - m $ \Comment remove the resolved base pair from $s$
\State {$newState = state(ns,cs.t,cs.dist + 1, l , cs.rm + m ,  cs.ad, cs.sh )$}
\EndIf
\State $Q = Q \cup newState$
\EndIf
\EndFor 
\EndIf
\EndIf
\EndWhile
\State $finalPath = BP_1 + Sol.rm + Sol.sh + Sol.add + BP_2 $
\end{algorithmic}
\end{small}
\end{algorithm}

In algorithm \ref{bb} we describe a branch and bound algorithm for computing the $MS_2$ distance between two RNA secondary stuctures $s$ and $t$. First base pairs in $BP_1$ and $BP_3$ are removed from $s$ to obtain the root structure.  Starting from the root we perform removal and shift of base pairs until the target or an empty structure is achieved. For each state $S_i$, we have the distance $dist$ from the root to $S_i$. Furthermore, we compute a lower bound $lb$ for the $MS_2$ distance from $S_i$ to the target structure by calculating the $MS_2$ distance allowing pseudoknots. The optimistic $MS_2$ distance from the source to target is estimated as $dist+lb$. If the optimistic distance is higher than the incumbent $MS_2$ distance computed so far, the search is stopped and a different state is considered. In order to prune the search tree more effectively, base pairs that cause more conflicts are considered first. Finally, if an empty structure is achieved we compute the list of remaining moves for $s$.

\section{Greedy algorithm}
\label{subsection:greedyAlgorithm}

For a digraph $G=(V,E)$, in this section, we present the pseudocode for 
a staightforward greedy algorithm to determine a (possibly non-maximal) 
vertex subset $\Vbar \subset V$ such that the induced subgraph
$H=(\Vbar,\Ebar)$ contains no directed cycles, where $\Ebar = E \cap
(\Vbar \times \Vbar)$. Nevertheless, in the following greedy algorithm, it
is necessary to first generate a list of all (possibly exponentially many)
directed cycles. This computational overhead is sidestepped by the 
near-optimal algorithm in the next section.

\begin{algorithmPeter}[Greedy approximation of $MS_2$ distance from $s$ to $t$] 
\label{algo:greedyMS2path}
\hfill\break
{\sc Input:} Secondary structures $s,t$ for RNA sequence $a_1,\ldots,a_n$
\hfill\break
{\sc Output:} Greedy $MS_2$ folding trajectory 
$s = s_0,s_1,\ldots,s_m = t$, where $s_0,\ldots,s_m$ are
secondary structures, $m$ is the minimum possible value for which
$s_{i}$ is obtained from $s_{i-1}$ by a single base pair addition, removal or
shift for each $i=1,\ldots,m$.
\end{algorithmPeter}
First, initialize the variable {\tt numMoves} to $0$, and the list
{\tt moveSequence} to the empty list {\tt [~ ]}. Define
$BP_1 = \{ (x,y) : (x,y) \in t, (t-s)[x]=0, (t-s)[y]=0\}$; i.e.
$BP_1$ consists of those base pairs in $s$ which are not touched by any
base pair in $t$. Define
$BP_2 = \{ (x,y) : (x,y) \in t, (s-t)[x]=0, (s-t)[y]=0\}$; i.e.
$BP_2$ consists of those base pairs in $t$ which are not touched by any
base pair in $s$.
Bear in mind that $s$ is constantly being updated, so actions performed
on $s$ depend on its current value.
\bigskip
\begin{small}
\mverbatim
  //remove base pairs from $s$ that are untouched by $t$
 1. for $(x,y) \in BP_1$
 2.   remove $(x,y)$ from $s$; numMoves = numMoves+1
  //define conflict digraph $G=(V,E)$ on updated $s$ and unchanged $t$
 3. define $V$ by equation (\ref{eqn:conflictDigraphVertexSetV})
 4. define $E$ by equation (\ref{eqn:conflictDigraphEdgeSetE})
 5. define conflict digraph $G=(V,E)$ 
 6. $\mathcal{C} = \{ C_1,\ldots,C_m \}$ //list of all simple directed cycles in $G$
  //determine set $V_0$ of vertices to remove so that restriction of $G$ to $V-V_0$ is acyclic
 7. $V_0 = \emptyset$  //$V_0$ is set of vertices to be removed from $V$
 8. for $v \in V$
 9.   $\mathcal{C}_v =  \{ C \in \mathcal{C}: v \in C \}$
10. while $\mathcal{C} \ne \emptyset$
11.   $v_0 = \mbox{argmax}_v ||\mathcal{C}_v||$ //$v_0$ belongs to largest number of cycles
12.   $V_0 = V_0 \cup \{ v_0 \}$
13.   $V = V - \{ v_0 \}$
14.   $E = E - \{ (x,y): x=v_0 \lor y=v_0 \}$
15.   $G=(V,E)$ //induced subgraph obtained by removing $v_0$
16.   $\mathcal{C} = \mathcal{C} - \mathcal{C}_{v_0}$ // remove all cycles containing $v_0$ 
17.   $v_0 = (x,y,z)$ //unpack $v_0$ to obtain base pairs $\{x,y\}_< \in t$, $\{ y,z\}_< \in s$
18.   $s = s - \{ (\min(y,z),\max(y,z)) \}$
  //topological sort of the now acyclic digraph $G=(V,E)$ for updated $V,E$
19. topological sort of $G$ using DFS \cite{cormen} to obtain total ordering $\prec$ on $V$
20. for $v=(x,y,z) \in V$ in topologically sorted order $\prec$
       //check if shift would create a base triple, as in type 1,5 paths from Figure 3 of main text
21.   if $s[x]=1$ //i.e. $\{u,x\} \in s$ for some $u\in [1,n]$
22.     remove $\{ u,x \}$ from $s$; numMoves = numMoves+1 
23.   shift $\{y,z\}$ to $\{x,y\}$ in $s$; numMoves = numMoves+1
  //remove any remaining base pairs from $s$ that have not been shifted
24. for $(x,y)\in s-t$
25.   remove $(x,y)$ from $s$; numMoves = numMoves+1
      //add remaining base pairs from $t-s$, e.g. from $BP_2$ and type 4,5 paths in Figure 3 of main text
26. for $(x,y)\in t-s$
27.   add $(x,y)$ to $s$; numMoves = numMoves+1
28. return folding trajectory, numMoves
|mendverbatim
\end{small}
\bigskip

We now analyze the time and space complexity of the greedy algorithm.
In line 6, Johnson's algorithm \cite{Johnson:cycleEnumeration}
is used to enumerate all simple 
directed cycles, resulting in run time $O((|V|+|E|) \cdot (|\mathcal{C}|+1) )$, 
where $|V|$ [resp. $|E|$] denotes the number of vertices [resp. edges] of 
the initial conflict digraph $G$, and $|\mathcal{C}|$ denotes the number 
of directed cycles of $G$. Let  $M = |\mathcal{C}|$ denote the number of
directed cycles in $\mathcal{C}$, and let $N=O(|V| \cdot M)$
denote the total number of vertices
(counting duplicates) in the set of all simple directed cycles
$\mathcal{C} = \{ C_1,\ldots,C_M \}$. 
Lines 7 through 28 require $O(N)$ time and space,
provided that one introduces the data structures $A_1,A_2,A_3,A_4$, defined
by 
as follows:
\begin{align*}
A_1[v] &= |\{ C \in \mathcal{C}: v \in C \}| \\
A_2[v] &= \{ k \in \{ 1,\ldots,|\mathcal{C}|\} : C_k \in \mathcal{C} \land v \in C_k \}\\
A_3[k]  &= \{ v \in V: v \in C_k \}\\
A_4[k]  &= \left\{ \begin{array}{ll}
1 &\mbox{if $C_k \in \mathcal{C}$}\\
0 &\mbox{else} \end{array} \right. 
\end{align*}
In other words, $A_1$ is a linked list of size $|V|$, where $A_1[v]$ equals the
(current) number of cycles to which $v$ belongs (in line 13, the node
$A_1[v]$ is deleted from the linked list);
$A_2$ is an array of size $|V|$, where $A_2[v]$ is a linked list of indices
$k$ of cycles $C_k$ that contain vertex $v$ (note that the size of linked
list $A_2[i]$ is $A_1[i]$);
$A_3$ is an array of size the number  $|\mathcal{C}|$ of cycles, where
$A_3[k]$ is a linked list of vertices $v$ that belong to $C_k$;
$A_4$ is an array of size the number  $|\mathcal{C}|$ of cycles, where
$A_4[k]$ is a a boolean value (true/false), depending on whether the cycle $C_k$
currently belongs to $\mathcal{C}$ (used to implement line 16). Details
are left to the reader, or can be gleaned from reading our publicly
available source code. It follows that the run time complexity of 
Algorithm~\ref{algo:greedyMS2path} is 
$O((|V|+|E|) \cdot (|\mathcal{C}|+1) )$ with space complexity of
$O(|V|\cdot (|\mathcal{C}|+1) + |E|)$.

\section{Graph theoretical properties}
\label{section:graphTheoreticalProperties}

\subsection{Representable digraphs}
\label{section:representableDigraphs}

Recall that digraph $G=(V,E)$ is {\em isomorphic} to digraph 
$G'=(V',E')$ if there is a bijective function (i.e. one-one and onto)
$\Phi:V \rightarrow V'$,
such that for all $u,v \in V$, $(u,v) \in E$ if and only if 
$(\Phi(u),\Phi(v))\in E'$. Since RNA conflict digraphs have a natural
ordering of vertices defined in Definition~\ref{def:digraphNode}, we
now define {\em digraph order-isomorphism}.

\begin{definition} [Order-isomorphism]
\label{def:orderIsomorphism}
Let $G=(V,E,\preceq)$ [resp. $G'=(V',E',\preceq')$] be a
digraph, whose vertex set $V$ [resp. $V'$] is totally ordered by
$\preceq$ [resp. $\preceq'$]. We say that $G$ is order-isomorphic to 
$G'$ if there exists
an order-preserving bijective function $\Phi: V \rightarrow V'$ 
(i.e. one-one and onto) such that 
(1) for $u,v \in V$, $x \preceq y$ if and only if $\Phi(u) \preceq' \Phi(v)$,
(2) for $u,v \in V$, $(u,v) \in E)$ if and only if $(\Phi(u),\Phi(v)) \in E'$.
If $\Phi$ is an injective function (one-one, but not necessarily onto), then
$G$ is said to have an order-preserving embedding in $G'$.
\end{definition}
We say that a digraph $G=(V,E)$ is {\em representable} if it is 
order-isomorphic to an RNA conflict digraph, formally defined as follows.
\begin{definition} [Representable digraph] \hfill\break
\label{def:representableDigraph}
Let $V = \{ 1,\ldots,n\}$ be a set of vertices and $E$ a set of directed
edges on $V$.  The digraph $G=(V,E)$ is said to be {\em representable} if
there exist secondary 
structures $s,t$ of some RNA sequence $a_1,\ldots,a_m$, an integer $N$, and
an order-preserving function $\Phi: [1,n] \rightarrow [1,N]^3$ such that 
(1) for $v,v' \in [1,n]$, $x<y$ if and only if $\Phi(v)<\Phi(v')$,
(2) for each $v \in [1,n]$, $\Phi(v) = (x,y,z)$ where
$x,y,z$ are distinct, $\{ x,y\}_< \in t$, $\{ y,z\}_< \in s$, 
(3) there is an edge $u \rightarrow v$ in $E$ if and only if 
$\Phi(u).s=\{ y,z\}_< \in s$ touches or crosses $\Phi(v).t = \{x,y\}_< \in t$. 
\end{definition}
As just defined, the notion of representability depends on the nucleotide
sequence $a_1,\ldots,a_n$. In a mathematical investigation to determine 
which digraphs are representable, it is more natural to reinterpret 
the notion of secondary structure to satisfy requirements 2-4 of 
Definition~\ref{def:secStr}, but not necessarily requirement 1.

The requirement that mapping $\Phi$ be order-preserving is important.
Consider the RNA conflict digraph $G$ in Figure~\ref{fig:4cycle}, equivalent
to the ordered digraph in 
Figure~\ref{fig:digraphsNotRepresentableAsConflictDigraphs}a, having edges
$1 \rightarrow 2 \rightarrow 3 \rightarrow 4 \rightarrow 1$.
Clearly $G$ is isomorphic to the digraph $G'$ in 
Figure~\ref{fig:digraphsNotRepresentableAsConflictDigraphs}b, although there is no order-isomorphism
between $G$ and $G'$. Indeed, by writing a program to exhaustively
enumerate all representable digraphs having a vertex set of size 4, we
know that $G'$ is not order-isomorphic to any RNA conflict digraph.
It is a straightforward exercise to show that
each of the $2^{ {3 \choose 2} } = 8$ many tounaments on 3 nodes is
representable (data not shown); however, not all
$2^{ {4 \choose 2} } = 64$ many tournaments on 4 nodes are representable,
as shown in Figure~\ref{fig:digraphsNotRepresentableAsConflictDigraphs}c. 
Although representability is
not invariant under isomorphism, it clearly is invariant under 
order-isomorphism. Moreover, we have the following.

\begin{theorem}
\label{thm:nonrepresentabilityInvariantUnderOrderPreservingEmbedding}
Suppose that $\Phi$ is an order-preserving embedding of digraph $G=(V,E)$
into digraph $G'=(V',E')$. If $G$ is not representable, then 
$G'$ is not realizable.
\end{theorem}
The theorem is immediate, since if $G'$ were order-isomorphic to an RNA
conflict digraph, then the induced subgraph $\Phi(G)$ of $G'$ must be
representable, and hence $G$ must be representable.
Figure~\ref{fig:digraphsNotRepresentableAsConflictDigraphsBis}a 
depicts a nonrepresentable digraph having
4 vertices and 4 edges. By adding an edge to that figure, we obtain the
digraph in Figure~\ref{fig:digraphsNotRepresentableAsConflictDigraphsBis}b, 
which is {\em not} representable.

Recall that an {\em automorphism} of a directed graph $G=(V,E)$ is the
set of permutations $\sigma$ on $n$ letters, for $V=\{1,\ldots,n\}$, such
that $G$ and $\sigma(G)$ are isomorphic. Using a small program that we
wrote to compute the automorphism group $Aut(G)$ for any connected, 
directed graph $G=(V,E)$, we found that the digraphs in
Figures~\ref{fig:digraphsNotRepresentableAsConflictDigraphs}c
and \ref{fig:digraphsNotRepresentableAsConflictDigraphsBis}b both
have the trivial automorphism group consisting only of the identity
permutation on 4 letters. Since the former is {\em not} representable and
the latter {\em is} representable, it follows that the automorphism group
of a digraph implies nothing about whether the digraph is representable.

\subsubsection{Example of RNA conflict digraphs}
\label{section:examples}

Computation time for the IP algorithm \ref{algo:MS2path} 
is dominated by the time to generate a list of all simple cycles and the
time to obtain an IP solution satisfying the FVS problem ($\dag$) as well
the constraints ($\ddag$) that ensure that shift moves cannot be applied 
if they share same base pair from $s$ or $t$. This raises the question 
whether the FVS problem is polynomial time solvable for RNA conflict digraphs.
We cannot settle this open question, but provide examples of RNA conflict 
digraphs that indicate the there is no reduction to known cases for which
the FVS problem has been resolved to be NP-complete or polynomial time 
computable.
Figure~\ref{fig:4cycle} depicts a 4-cycle digraph that is order-isomorphic
to the digraph with edges
$1 \rightarrow 2 \rightarrow 3 \rightarrow 4 \rightarrow 1$.
Figure~\ref{fig:forbiddenGraphs}a 
shows the complete bipartite graph $K_{3,3}$.
Recall that in \cite{forbiddengraphK33planar}, Kuratowski proved
that a graph is planar if and only if it does not contain a subgraph
that is a subdivision of $K_{3,3}$ or the complete graph $K_5$; i.e.
does not contain an embedded copy of one of the forbidden graphs
$K_{3,3}$ or $K_5$. 
Figure~\ref{fig:forbiddenGraphs}b 
depicts a copy of $K_{3,3}$ with directed edges,
which is realized in the RNA conflict digraph shown in
Figure~\ref{fig:RNAconflictGraphForK33}. 
It follows from Kuratowski's theorem that
RNA conflict digraphs are not planar in general.
Figure~\ref{fig:forbiddenGraphs}c 
depicts a forbidden digraph, with the property that
a flow graph (graphical representation of code in a programming language
with possible {\sc goto}-statements) is {\em reducible} if and only if
the flow graph does not contain a copy of this forbidden digraph (or
related digraph where the edges of Figure~\ref{fig:forbiddenGraphs}c 
may be replaced by {\em arc-disjoint} paths) \cite{flowgraphUllman}.
Figure~\ref{fig:RNAconflictGraphForFlowGraph} shows an RNA conflict digraph
which represents the forbidden digraph of Figure~\ref{fig:forbiddenGraphs}c.
It follows that RNA conflict digraphs are not reducible flow graphs
in general.
Two classes of digraphs for which FVS
is known to be polynomial time computable are: (1) planar digraphs
\cite{planarMinimaxArcTheorem}, and (2) reducible flow graphs
\cite{ramachandranMinimaxArcTheorem}. On the other hand, FVS
is NP-complete for general digraphs  \cite{karpNPcomplete}, for 
tournaments 
\cite{tournamentFeedbackArcNPcomplete,bookDigraphdirectedgraph}, and for
Eulerian digraphs \cite{feedbackarcsetNPcompleteForEulerianDigraphs}.
It is a simple exercise, left to the reader,  
to show that each of the $2^{{3 \choose 2}}=8$
tournaments on 3 vertices can be represented by an RNA conflict digraph.
\begin{proposition}
All tournaments on a graph having 3 vertices can be represented as
RNA conflict digraphs.
\end{proposition}
There are $2^{ {n \choose 2}}$ many tournaments on $n$ vertices.
By applying constraint programming to each of the $2^6=64$ tournaments
on 4 vertices, we found a number of tournaments that are not representable
as conflict digraphs. 
Figure~\ref{fig:digraphsNotRepresentableAsConflictDigraphs}a 
depicts a 4-cycle which is representable by
the RNA conflict digraph represented in Figure~\ref{fig:4cycle}.
Figure~\ref{fig:digraphsNotRepresentableAsConflictDigraphs}b
depicts a simple, connected digraph on
4 nodes (a cycle) which is not representable as a conflict digraph, and
Figure~\ref{fig:digraphsNotRepresentableAsConflictDigraphs}c 
shows a tournament on 4 nodes which is not
representable. 

The FVS problem is known to be NP-complete for
all tournaments, all Eulerian digraphs and for general digraphs, and 
there are polynomial time algorithms for FVS for planar digraphs and for 
flow graphs. Examples we have provided suggest that there is no 
straightforward application of known results to settle the question whether
the  FVS problem is NP-complete for
the class of RNA conflict digraphs.

\begin{question}
\label{question:RNAconflictDigraphFeedbackArcPolytimeOrNPcomplete}
Is the FVS problem polynomial time computable or NP-complete for
the collection of RNA conflict digraphs?
\end{question}
Using constraint programming, we determined which connected digraphs on 4 nodes
could be represented as RNA conflict digraphs. This was done by considering
all partition of $[1,12]$ into four classes, each class corresponding to 
a triplet node $(x,y,z)$ and determining the resulting edge relations defined
by whether a base pair from $s$ belonging to node $v$ crosses a base pair
from $t$ belonging to node $v'$. Is there a better approach?
\begin{question}[Representation of arbitrary $G=(V,E)$ as conflict digraph]
\label{question:representationAsConflictDigraph}
Is there an efficient algorithm to determine whether the labeled digraph
$G=(V,E)$ can be realized as an RNA conflict digraph, where $G=(V,E)$ is given
as input.
\end{question}
The following is perhaps tractable.
\begin{question}
\label{question:representation2}
Is there an efficient algorithm to determine whether a given labeled digraph
$G=(V,E)$ can be realized as an RNA conflict digraph, 
the vertices of $G$ are totally ordered by $v_1 < \cdots < v_n$,
and each vertex is labeled by one of the node types 1,\ldots,6.
\end{question}
Even if the preceding problem has a polynomial time solution, it is not
clear whether the same is true for the following slight generalizations.
\begin{question}
\label{question:representation3}
Is there an efficient algorithm to determine whether the labeled digraph
$G=(V,E)$ can be realized as an RNA conflict digraph, 
where vertices of $D$ are totally ordered by $v_1 < \cdots < v_n$.
\end{question}
\begin{question}
\label{question:representation4}
Is the feedback arc set (FAS) problem complete for RNA conflict digraphs?
Is the problem of computing $MS_2$ distance between arbitrary secondary
structures $s,t$ NP-complete?
\end{question}
Since a digraph $G$ is planar if and only if it contains neither the complete 
graph $K_5$ on 5 vertices nor the complete bipartite graph $K_{3,3}$, so
Theorem~\ref{thm:nonrepresentabilityInvariantUnderOrderPreservingEmbedding}
suggests the following question.

\begin{question}
Is there a finite set of prohibited digraphs $H_1,\ldots,H_r$ such that
a digraph $G$ is representable by an RNA conflict digraph if and only if
there is no order-preserving embedding of $H_i$ into $G$ for any $i=1,\ldots,r$.
\end{question}